\documentclass[10pt,fleqn]{article}
\pdfoutput=1
\usepackage{psfrag}
\usepackage{amsmath}
\usepackage[dvips]{epsfig}
\usepackage{epsfig}
\usepackage{epstopdf}
\usepackage{setspace}
\usepackage{amsmath,amssymb,bm}
\usepackage{graphics}
\usepackage[pdftex,colorlinks,bookmarks,pdfpagemode=UseOutlines,linkcolor=black,
urlcolor=black,citecolor=black,pageanchor=false]{hyperref}
\usepackage[margin=1in]{geometry}
\usepackage[T1]{fontenc}
\usepackage[authoryear,round]{natbib}
\usepackage[pagewise,mathlines]{lineno}
\usepackage{amsthm}
\usepackage{longtable}
\usepackage{xfrac}
\usepackage{subcaption}
\usepackage{multirow}
\usepackage{authblk}

\newtheorem{theo}{Theorem}[section]
\newtheorem{lemma}{Lemma}[section]
\newtheorem{df}{Definition}[section]
\newtheorem{cor}{Corollary}[section]
\newtheorem{assump}{Assumption}[section]
\newtheorem{assert}{Assertion}[section]
\newtheorem{remark}{Remark}[section]
\newcommand{\bl}{\begin{lemma}}
\newcommand{\el}{\end{lemma}}
\newcommand{\be}{\begin{equation}}
\newcommand{\ee}{\end{equation}}
\newcommand{\beqn}{\begin{eqnarray}}
\newcommand{\eeqn}{\end{eqnarray}}
\newcommand{\bt}{\begin{theo}}
\newcommand{\et}{\end{theo}}
\newcommand{\bd}{\begin{df}}
\newcommand{\ed}{\end{df}}
\newcommand{\ba}{\begin{assump}}
\newcommand{\ea}{\end{assump}}
\newcommand{\bass}{\begin{assert}}
\newcommand{\eass}{\end{assert}}
\newcommand{\brem}{\begin{remark}}
\newcommand{\erem}{\end{remark}}
\newcommand{\bc}{\begin{cor}}
\newcommand{\ec}{\end{cor}}

\newcommand{\CA}{{\cal A}}

\newcommand{\QQ}{{\cal Q}}

\newcommand{\pt}{\tilde{p}}
\newcommand{\Rt}{\tilde{R}}
\newcommand{\St}{\tilde{S}}

\newcommand{\Ut}{\tilde{U}}

\newcommand{\etab}{\bm{\eta}}
\newcommand{\tetab}{\tilde{\bm{\eta}}}

\numberwithin{equation}{section}

\long\def\comment#1{}

\title{On node models for high-dimensional road networks}
\author[1,2]{Matthew A. Wright}
\author[2]{Gabriel Gomes}
\author[1,2]{Roberto Horowitz}
\author[2]{Alex A. Kurzhanskiy}
\affil[1]{Department of Mechanical Engineering, University of California, Berkeley}
\affil[2]{Partners for Advanced Transportation Technologies (PATH), University of California, Berkeley}

\begin{document}

\maketitle

\begin{abstract}
Macroscopic traffic models are necessary for simulation and study of traffic's complex macro-scale dynamics, and are often used by practitioners for road network planning, integrated corridor management, and other applications.
These models have two parts: a link model, which describes traffic flow behavior on individual roads, and a node model, which describes behavior at road junctions.
As the road networks under study become larger and more complex --- nowadays often including arterial networks --- the node model becomes more important.
Despite their great importance to macroscopic models, however, only recently have node models had similar levels of attention as link models in the literature.
This paper focuses on the first order node model and has two main contributions.
First, we formalize the multi-commodity flow distribution at a junction as an optimization problem with all the necessary constraints.
Most interesting here is the formalization of input flow priorities.
Then, we discuss a very common ``conservation of turning fractions'' or ``first-in-first-out'' (FIFO) constraint, and how it often produces unrealistic spillback.
This spillback occurs when, at a diverge, a queue develops for a movement that only a few lanes service, but FIFO requires that all lanes experience spillback from this queue.
As we show, avoiding this unrealistic spillback while retaining FIFO in the node model requires complicated network topologies.
Our second contribution is a ``partial FIFO'' mechanism that avoids this unrealistic spillback, and a (first-order) node model and solution algorithm that incorporates this mechanism.
The partial FIFO mechanism is parameterized through intervals that describe how individual movements influence each other, can be intuitively described from physical lane geometry and turning movement rules, and allows tuning to describe a link as having anything between full FIFO and no FIFO.
Excepting the FIFO constraint, the present node model also fits within the well-established ``general class of first-order node models'' for multi-commodity flows.
Several illustrative examples are presented.
\end{abstract}


\section{Introduction}\label{sec_intro}
\defcitealias{csmps}{Caltrans, 2015}


Traffic simulation models are vital tools for traffic engineers and practitioners. As in other disciplines focusing on complex systems, such as climate or population dynamics, traffic models have helped to deepen our understanding of traffic behavior. They are widely used in transportation planning projects in which capital investments must be justified with simulation-based studies~\citepalias{csmps}. Recently, with the increased interest in Integrated Corridor Management (ICM) and Decision Support Systems (DSS), traffic models have also found a new role in real-time operations management. In both the planning and ICM/DSS contexts, the models have steadily grown larger - such as highway plans that append adjacent arterial networks or managed lanes~\citep{floridaDOT2013} - and more complex - such as integration of new ``smart'' vehicle and communications technologies into real-time ICM.

Traffic models are typically divided into three categories based on their level of abstraction. At the most granular level, microscopic models simulate the motion and behavior of individual vehicles. At the other extreme, macroscopic models describe the evolution of traffic flows and buildup and breakdown of congestion along the lineal direction of a road in aggregate terms. Mesoscopic traffic models occupy the intermediate space. Of the three, the highly-abstracted macroscopic models unsurprisingly have the lowest computational cost, which makes them well-suited for study of large and complex networks of roads.

A macroscopic model is said to consist of a \emph{link} model and a \emph{node} model.
A link model describes the evolution through time of the traffic flow along homogeneous sections of road.
Several types of link models exist, such as those that give link flows as functions of link densities (e.g., the widely-used cell transmission model of \citet{daganzo94} and its descendants), those that track the vehicles only at link boundaries \citep{yperman_link_2005}, and others (see \citet{nie_linkmodels_2005} for a broad overview).
For the purposes of this paper, we do not specify a particular class of link model, but only require that the model describe, as a function of its state at time $t$, both the amount of vehicles trying to exit the link ($S(t)$, the link's \emph{demand}) and the amount of vehicles that the link is able to accept from upstream ($R(t)$, the link's \emph{supply}).

A traffic model may contain multiple classes of vehicles that share the road, and each class may have their own demands.
Separate vehicle classes are often called \emph{commodities}.
In addition to the link and node models are the so-called \emph{turning} or \emph{split ratios}, which define the vehicles' turning choice at the junction --- the ratios of vehicles of each commodity that take each of the available movements.
These split ratios might be measured, such as by manually counting flows of each movement at traffic intersections, estimated from some model, or prescribed by the modeler to make vehicles follow a certain path.
In the context of this paper, we consider only the split ratios at particular junctions.
Nodes join the links, and the node model computes the set of flows through a node for each commodity as a function of its incoming links' demands and its outgoing links' supplies.

Of course, while macroscopic models may be fast relative to more granular meso- and microscopic models, their computational needs are affected with the growth of network size and complexity.
High computational cost can be exacerbated in modern ICM applications, as well.
Many real-time traffic state estimation techniques follow an \emph{ensemble method} approach, where many simulations describing different possible events are processed simultaneously (see e.g.~\citet{work2009trafficmodel},~\citet{wright_pf_2016}).
ICM decision-making can follow a similar approach, with multiple simulations being performed at a plan-evaluation step to project traffic outcomes under a range of possible future demands.

While the computational complexities stemming from links have been well-studied, node models can be sources of computational costs as well.
These costs emerge when one models a network with many multi-input and/or multi-output junctions, or junctions where more than two links enter or exit.
We adopt the term ``high-dimensional'' to describe these sorts of networks, to avoid ambiguity with similar terms such as ``large,'' which may also describe networks with many long roads and not many junctions.
This paper draws on the authors' experience in creating models for high-dimensional networks that describe a freeway, adjacent managed lanes (such as high-occupancy vehicle (HOV) lanes or tolled lanes), and/or the surrounding arterial grid for ICM purposes.
We will discuss in Section \ref{sec_node_review} how modern node models can exacerbate computational complexity in high-dimensional networks by creating a trade-off between model accuracy and number of links; thankfully, this trade-off can be overcome in a simple manner, as we will detail in Section~\ref{sec_node_model}.
Section \ref{sec_examples} gives two examples, one being a slight modification of a well-studied example from \citet{tampere11}, in which we demonstrate how our new node model considerations differ from previous node models in the junction-geometry information they consider and the flows produced.
Appendix~\ref{app_notation} summarizes the notation used in this paper.

\section{Common node models and their drawbacks}\label{sec_node_review}
\subsection{Node models}\label{subsec_node_review}
The traffic node problem is defined on a junction of $M$ input links, indexed by $i$,
and $N$ output links, indexed by $j$, with $C$ vehicle commodities, indexed by $c$.
As mentioned above, in first-order traffic models, the node model is said to
consist of the mechanism by which, at time $t$, incoming links' per-commodity demands $S_i^c(t)$, split ratios
$\beta_{i,j}^c(t)$ (which define the portion of vehicles of commodity $c$ in link $i$ that wish
to exit to link $j$), and outgoing
links' supplies $R_j(t)$ are resolved to produce throughflows $f_{i,j}^c(t)$.
Nodes are generally infinitesimally small and have no storage, so all the flow that enters the node at time $t$ must exit at time $t$.
To simplify the
notation, in the remainder of this paper we consider the node model evaluation in each time instant as an
isolated problem (as we are assuming no storage), and omit the variable $t$ for these quantities.

The node problem's
history begins with the original formulation of discretized first-order traffic flow
models~\citep{daganzo95a}. There have been many developments in the node model theory
since, but we will reflect only on some more recent results.

We can divide the node model literature into pre- and
post-\citet{tampere11} epochs. \citet{tampere11} drew from the literature several
earlier-proposed node model requirements to develop a set of necessary
conditions for first-order node models that they call the ``general class of first-order node models.''
To review, these conditions are:
\begin{enumerate}
	\item Applicability to general numbers of input links $M$ and
	output links $N$. In the case of multi-commodity flow, this requirement
	also extends to general numbers of commodities $c$.
	
	\item Maximization of the total flow through the node.
	Mathematically, this may be expressed as $\max \sum_{i,j,c} f_{i,j}^c$.
	According to~\citet{tampere11}, this means that
	``each flow should be actively restricted by one of the constraints,
	otherwise it would increase until it hits some constraint.''
	When a node model is formulated as a constrained optimization problem,
	its solution will automatically satisfy this requirement.
	However, what this requirement really means is that constraints should be
	stated \emph{correctly} and not be overly simplified and, thus,
	overly restrictive for the sake of convenient problem formulation.
	See the literature review in \citet{tampere11} for examples of node models
	that inadvertently do not maximize node throughput by oversimplifying their
	requirements. \label{item:maxFlow}
	
	\item Non-negativity of all input-output flows. Mathematically, $f_{i,j}^c \ge 0$
	for all $i,j,c$.
	
	\item Flow conservation: Total flow entering the node must be equal to total flow
	exiting the node. Mathematically, $\sum_i f_{i,j}^c = \sum_j f_{i,j}^c$ for all $c$.
	\item Satisfaction of demand and supply constraints. Mathematically, $\sum_j f_{i,j}^c \le S_i^c$ and $\sum_i f_{i,j}^c \le R_j$.
	
	\item Satisfaction of the first-in-first-out (FIFO) constraint: if a single
	destination $j$ for a given $i$ is not able to accept all demand
	from $i$ to $j$, all flows from $i$ are also constrained. This requirement is
	sometimes called ``conservation of turning fractions'' (CTF). Mathematically,
	$f_{i,j}^c / \sum_i f_{i,j}^c = \beta_{i,j}^c$.
	
	We believe that in some situations, the FIFO constraint may be too restrictive.
	It should not be completely eliminated, but rather, could be relaxed through
	a parametrization. We will revisit this point in the following Section.
	
	\item Satisfaction of the invariance principle.
	If the flow from some input link $i$ is restricted by the available output
	supply, this input link enters a congested regime.
	This creates a queue in this input link and causes its demand $S_i$ to jump
	to the link's capacity, which we denote $F_i$, in an infinitesimal time, and therefore,
	a node model should yield solutions that are invariant to replacing $S_i$ with
	$F_i$ when flow from input link $i$ is supply-constrained~\citep{lebacque05}.
\end{enumerate}

This paper concerns multi-commodity traffic, which is unaddressed in the preceding
list. To address multi-commodity traffic, we add another requirement,
\begin{enumerate}
\item[8.] Supply restrictions on a flow from any given input link are
imposed on commodity components of this flow proportionally to their 
per-commodity demands.
\end{enumerate}
Requirement 8 assumes that the commodities are mixed isotropically.
This means that all vehicles attempting to take movement $i,j$ will be queued in roughly random order, and not, for example, having all vehicles of commodity $c=1$ queued in front of all vehicles of $c=2$, in which case the $c=2$ vehicles would be disproportionally affected by spillback.
We feel this is a reasonable assumption for situations where the demand at the node is dependent mainly on the vehicles near the tail of the link (e.g., in a small cell at the tail).

In addition to the above numbered requirements, two other elements are required
to define a node model. The first is a rule for the portioning of output link supplies
$R_j$ among the input links. Following~\citet{gentile07}, in~\citet{tampere11} it was 
proposed to allocate supply for incoming flows proportionally to input link capacities $F_i$.

The second necessary element is a redistribution of ``leftover supply.'' Following the
initial partitioning of supplies $R_j$, if one or more of the supply-receiving input links
does not fill its allocated supply,
some rule must redistribute the difference to other input links who may still fill it.
This second element is meant to model the selfish behavior of drivers to take any space
available, and ties in closely with requirement \ref{item:maxFlow} above.
\citet{tampere11} referred to these two elements collectively as a ``supply
constraint interaction rule'' (SCIR).

Mentioned as an optional addition by \citet{tampere11} are ``node supply constraints.''
What is meant by these are supply quantities internal to the node, in addition
to link supplies $R_j$. These node supplies are meant to model \emph{shared resources}
that multiple movements may make use of: each movement $i,j$ through the node may
or may not consume
an amount of a node supply proportional to $\sum_c f_{i,j}^c$. A node supply being
fully consumed would constrain these flows in a manner analogous to the link
supplies. The obvious example of such a ``shared resource'' node supply is green
time at a signalized intersection; see \citet{tampere11} for further discussion
of this example. In \citet{corthout_non-unique_2012} it was noted that these 
node supplies may lead to non-unique solutions.
Very recently, \citet{jabari_node_2016} revisited this non-uniqueness problem,
introduced a relaxation of the flow-maximization objective termed a holding-free flow objective
(i.e., not all holding-free flows are flow-maximizing) and proposed adding constraints
concerning signal timing and the conflicts between movements (e.g., a through movement conflicts
with the opposing left turn movement and if traffic drives on the right side of the road and the two
cannot send vehicles at the same time) to achieve a unique flow solution.
We will not make use of node supplies in the remainder of this paper.

The \citet{tampere11} requirements invite the modeler to design supply portioning
and/or redistribution rules that, in some sense, model traffic behavior, add them 
to the numbered requirements, and create new node models. Recall that \citet{tampere11}
suggested portioning the supplies $R_j$ proportionally to the input links capacities $F_i$.
\citet{gibb_model_2011} suggested a supply portioning rule based on a so-called
``capacity-consumption equivalence,'' which states, qualitatively, that movements
into a supply-constrained output link $j'$ are characterized by drivers from individual
input links $i$ taking their movements $i,j'$ one after the other in turn. The supply
portioning rule thus says that the input links $i$, where drivers spend more time
waiting for entry into congested links $j$, are assigned less supply from all links,
including those not in congestion.

\citet{smits_family_2015} argue that the quantities of interest in the node model
problem should not be the flows $f_{i,j}^c$, but rather the quantities
$1/ ( \sum_c f_{i,j}^c)$, which are equal to the amount of time (relative to the
simulation timestep) that a vehicle needs to make movement $i,j$ through the node.
They claim that this quantity is measurable in the field, and good node models
would produce flows that emulate observed movement times. They review the node
models proposed in \citet{tampere11} and \citet{gibb_model_2011} in this
re-parametrization and present others with different supply portioning and
redistribution rules.

It is important to note that the node model in
\citet{gibb_model_2011} and the so-called ``equal delay at outlink'' model of
\citet{smits_family_2015} define their supply portioning and redistribution
rules \emph{implicitly}. A consequence is that their solutions must be found
through a fixed-point iterative procedure, whose iterations-until-convergence
are potentially unbounded (this point was originally raised by \citet{smits_family_2015}).
On the other hand, the capacity-proportional node model of
\citet{tampere11} and ``single server'' model of \citet{smits_family_2015} 
are examples of \emph{explicit} supply portioning rules.

For our purposes, we will focus on node models
with explicit supply portioning and redistribution rules (that is, an explicit SCIR).
Our SCIR, which will be introduced in Section \ref{sec_node_model}, is similar to that of 
\citet{tampere11} and others in that it assigns leftover supply to other input links that
can make use of it, but differs in what is considered ``leftover'' and how it is assigned.

While \citet{tampere11} used the capacity-proportional supply portioning rule,
we will slightly generalize this to admit arbitrary per-input link priorities
$p_i$ in the spirit of \citet{daganzo95a}, \citet{ni_simplified_2005}, and 
\citet{flotterod_operational_2011}.

\subsection{Implications of the FIFO flow rule}
\label{subsec_fifo_implications}
We noted in Section~\ref{subsec_node_review} that the FIFO, or ``conservation of
turning fractions'' rule in the~\citet{tampere11} node requirements may be
too restrictive. Consider the example junction depicted in Figure~\ref{fig-simo}.
It depicts a five-lane freeway with left and right off-ramps. The natural network representation
of this junction would have one node with one input link, with demand $S_1$,
and three output links, with supplies $R_1$, $R_2$, and $R_3$, similar to
Figure~\ref{fig:1to3}.

\begin{figure}[htb]
\centering
\includegraphics[width=4in]{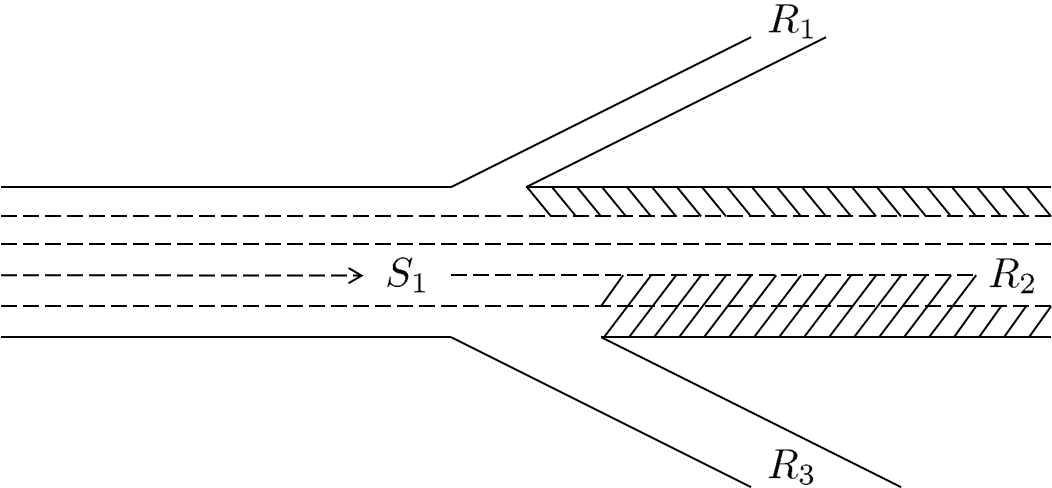}
\caption{A node with one input and three output links, where
congestion in output links~1 and~3 only partially affects
flow into output link~2, while congestion in link~2 affects
flows into output links~1 and~3 in full, and output links~1
and~3 do not affect each other.}
\label{fig-simo}
\end{figure}

\begin{figure}
\centering
	\begin{subfigure}[b]{.3\textwidth}
		\includegraphics[width=\textwidth]{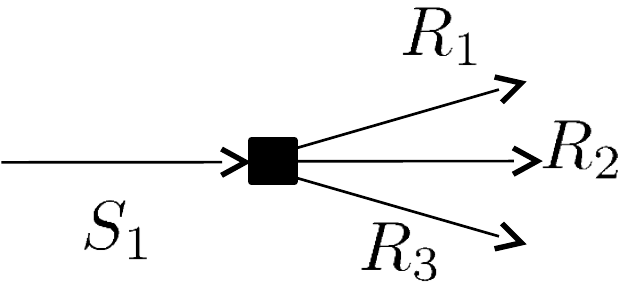}
		\caption{}
		\label{fig:1to3}
	\end{subfigure} \quad
	\begin{subfigure}[b]{.3\textwidth}
			\includegraphics[width=\textwidth]{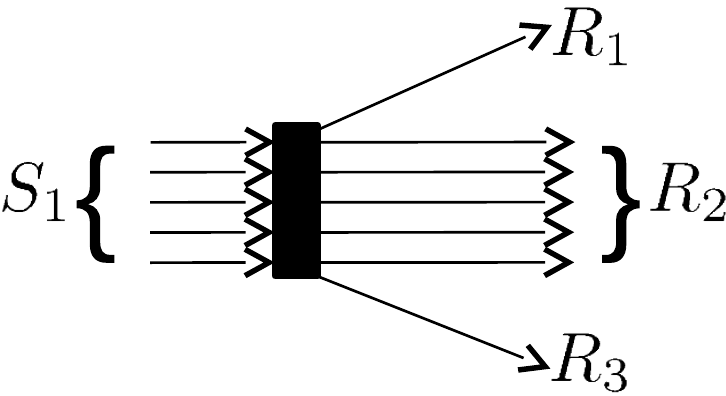}
			\caption{}
			\label{fig:seperateLanes}
	\end{subfigure} \quad
	\begin{subfigure}[b]{.3\textwidth}
			\includegraphics[width=\textwidth]{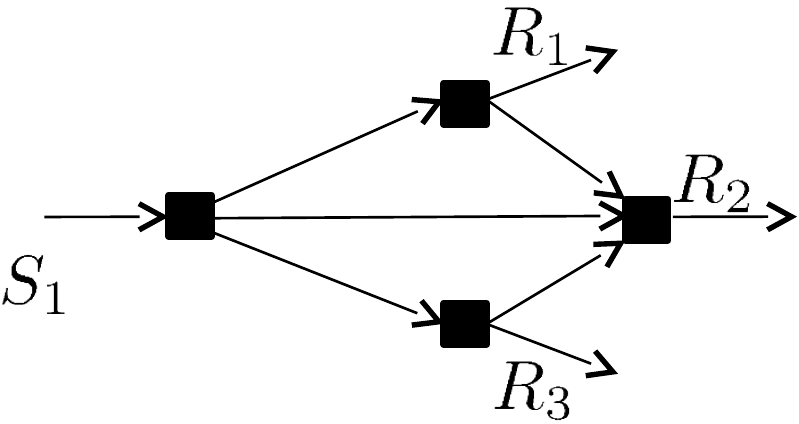}
			\caption{}
			\label{fig:intermediate}
	\end{subfigure}
	\caption{Three network representations of the junction in Figure~\ref{fig-simo}.}
	\label{fig:networks}
\end{figure}

Here, the FIFO constraint translates into
\begin{equation*}
\frac{f_{1,1}}{S_{1,1}} =
\frac{f_{1,2}}{S_{1,2}} =
\frac{f_{1,3}}{S_{1,3}},
\end{equation*}
where $f_{i,j} (S_{i,j})$ is the flow (demand) from link $i$ to $j$.
This equality means that if, for instance, off-ramp~1 is jammed
preventing traffic from coming in, 
there will be no flow of traffic to mainline~2 and off-ramp~3 either.
This is in contrast to the fact that if congestion on off-ramp~1 spilled back
onto the mainline, only one lane in Figure~\ref{fig-simo} would
be blocked (with some knock-on effects in the adjacent lane if drivers
were trying to enter the leftmost lane).
More generally, FIFO imposes that a jam in any of the output links will block the
flow to all output links.
Transportation engineers find that for certain junctions the
FIFO requirement is too strict to be realistic for this reason.

Some authors (e.g., \citet{bliemer07} or \citet{shiomi2015})
have suggested modeling each lane as a separate link to alleviate this
problem, similar to the network in Figure~\ref{fig:seperateLanes} (in
the Figure, the rectangular shape of the node signifies nothing, it 
is only stretched to fit all the input and output links). The demand $S_1$
and supply $R_2$ are split across the link lanes. This model
resolves the unrealistic spillback problem: congestion in off-ramps~1 or~3
will only spill back to the links that are sending demand to an off-ramp;
namely, the links representing the adjacent lanes. However, this approach
has its own drawbacks. First and most obviously, it greatly complicates the
size and dimensionality of the model, making every node in a two-or-more
lane road a merge and diverge. Second, it invites the question of how to
assign commodity traffic across link lanes, and provide a model for weaving or
lane-change behavior, questions that are abstracted away in the ``more macro''
macroscopic model of Figure~\ref{fig:seperateLanes}.

Figure~\ref{fig:intermediate} shows an intermediate approach --- the sending link
$S_1$ in Figure~\ref{fig-simo} is split upstream of the off-ramps, at the leftmost node
in Figure~\ref{fig:intermediate}.
The top and bottom links exiting this leftmost node represent the left and right lanes in
Figure~\ref{fig-simo}, and carry all exiting vehicles, as well as the vehicles that
take the left or right lanes but do not exit. The middle link then represents the middle
three lanes. The non-exiting vehicles rejoin the same link at the rightmost node.
This model limits the problems raised by separating each lane into a separate link,
as in the network of Figure~\ref{fig:seperateLanes}.
This approach, however, does not fully eliminate the unrealistic spillback problem, as heavy
spillback that congests the ``left lane'' or ``right lane'' link will still spill back
into link $S_1$. Splitting the link apart earlier would give us more room to avoid
the unrealistic spillback, but if we go too far upstream we can end up back at the network of
Figure~\ref{fig:seperateLanes}.

At this point it may appear that we have found ourselves modeling between Scylla and Charybdis: keeping the network at manageable simplicity invites unrealistic spillback behavior, while preventing the overly aggressive spillback causes the network to explode in size. In Section \ref{sec_node_model} we present a method to overcome this dilemma, and will revisit this example in the context of our proposed solution in Section \ref{subsec_simo}.
Our method is relatively simple: it relies only on the lane geometry in junctions, as in Figure~\ref{fig-simo}.
It is also general and could be used with any sort of network configuration,
for example, those of Figures~\ref{fig:seperateLanes} and~\ref{fig:intermediate}.

\section{A node model for dimensionality management}\label{sec_node_model}
Our notation for the node model problem has been introduced above.
For a quick reference, we refer the reader to Appendix \ref{app_notation}.
%

This Section is organized as follows.
First, we describe the merge problem --- the multi-input-single-output
(MISO) node in Section~\ref{subsec_miso} to explain the concept of input 
link priorities.
Then, we propose a way of relaxing the FIFO
condition and explain it in the case of a single-input-multi-output (SIMO) node
in Section~\ref{subsec_simo}.
Finally, in Section~\ref{subsec_mimo_relaxed} we proceed to the
general multi-input-multi-output (MIMO) node putting all the concepts
together.

\subsection{Multiple-Input-Single-Output (MISO) Node: Explaining Input Priorities}\label{subsec_miso}
We start by considering a node with $M$ input links and 1 output link.
The number of vehicles of type $c$ that input link $i$ wants to send is $S_i^c$.
The flow entering from input link $i$ has priority $p_i\geq 0$.
Here $i=1,\dots,M$ and $c=1,\dots,C$.
The output link can receive $R_1$ vehicles.

Let us briefly discuss the meaning of priorities.
One intuitive way to understand them is that qualitatively speaking, $p_i$, the priority of link $i$, can be thought of the ability of link $i$ to ``claim'' the supplies of the downstream links.
If we consider the node problem as describing a scenario where, over some period of time, drivers from each incoming link ``compete'' to claim downstream supply, then an intuitive understanding of priorities is that $p_i$ is proportional to the rate at which drivers are able to leave link $i$ and claim downstream supply.
We explore this idea further in a companion paper, \citet{wright_node_dynamic_2016}, where the node problems considered in this paper are explicitly recast as dynamic systems governed by ordinary differential equations.
Note also that this ``rate of leaving'' characterization of priorities was mentioned by \citet{tampere11}, in their justification of using capacities in place of priorities (their assumption being that vehicles will leave link $i$ at a rate equal to its capacity).
However, this characterization of priorities as flow rates is not meant to restrict the selection of priorities to flow rates.
What is important is the ratio of a $p_i$ to the other input links' priorities.
If, for example, a link 1 is to have twice the priority of a link 2, a priority set of $\{ 2000, 1000 \}$ (i.e., on the scale of flow rates in terms of veh/hour) will produce the same results as a priority set of $\{\sfrac{2}{3}, \sfrac{1}{3}\}$.

Casting the computation of input-output flows $f_{i1}^c$ as a mathematical
programming problem, we arrive at:
\be
\max\left(\sum_{i=1}^M\sum_{c=1}^C f_{i,1}^c\right), \label{miso_objective}
\ee
subject to:
\begin{eqnarray}
& & f_{i,1}^c \geq 0, \;\; i=1,\dots,M, \; c=1,\dots,C \;\;
\mbox{ --- non-negativity constraint};
\label{miso_nonnegativity_constraint} \\
& & f_{i,1}^c \leq S_i^c, \;\; i=1,\dots,M, \; c=1,\dots,C \;\;
\mbox{ --- demand constraint};
\label{miso_demand_constrint} \\
& & \sum_{i=1}^M f_{i,1} \leq R_1 \;\;
\mbox{ --- supply constraint};
\label{miso_supply_constraint} \\
& & \frac{f_{i,1}^c}{f_{i,1}} =
\frac{S_{i}^c}{S_{i}}, \;\;
i = 1,\dots,M,\; c=1,\dots,C,
\;\; \mbox{ --- proportionality constraint for commodity flows};
\label{miso_proportionality_constraint}\\
& & \left.\begin{array}{cl}
\mbox{(a)} &
p_{i''} f_{i',1} = p_{i'} f_{i'',1} \;\;
\forall i',i'', \mbox{ such that }
f_{i',1}<S_{i'}, \;
f_{i'',1}<S_{i''}, \\
\mbox{(b)} & 
\mbox{If } f_{i,1} < S_i, \mbox{ then }
f_{i,1} \geq \frac{p_i}{\sum_{i'=1}^M p_{i'}} R_1. 
\end{array}\right\} \;\; \mbox{ --- priority constraint}.
\label{miso_priority_constraint}
\end{eqnarray}
where $i'$ and $i''$ in \eqref{miso_priority_constraint} denote two
generic input links.
Constraint~\eqref{miso_proportionality_constraint} defines how 
potential restrictions imposed on full input-output flows are
translated to individual commodities.
In this paper we assume that any restriction on the
$i$-to-$1$ flow applies to commodity flows proportionally to
their contribution to the demand, in accordance with requirement 8
introduced in Section \ref{subsec_node_review}.

Let us discuss the priority constraint~\eqref{miso_priority_constraint}
in more detail.
It should be interpreted as follows.
Input links $i=1,\dots,M$, fall into two categories:
(1) those whose flow is restricted by the allocated supply of
the output link; and
(2) those whose demand is satisfied by the supply allocated for them
in the output link.
Priorities define how supply in the output link is allocated for input flows.
Condition~\eqref{miso_priority_constraint}(a)
says that flows from input links of category 1 are allocated
proportionally to their priorities.
Condition~\eqref{miso_priority_constraint}(b)
ensures that input
flows of category 1 get \emph{no less} than the portion of supply
allocated for them based on their priorities.
The inequality
in condition~\eqref{miso_priority_constraint}(b) becomes an
equality when category 2 is empty.

There is a special case when some input link priorities are
equal to $0$.
If there exists an input link $\hat{\imath}$ such that
$p_{\hat{\imath}}=0$, while $f_{\hat{\imath},1}>0$, then,
due to condition~\eqref{miso_priority_constraint}(a),
all input links with nonzero priorities are in category 2.
Thus, if category 1 contains only input links with zero priorities,
one should evaluate condition~\eqref{miso_priority_constraint}(a)
with arbitrary positive, but \emph{equal}, priorities: $p_{i'}=p_{i''}>0$.

If the priorities $p_i$ are proportional to the demands
$S_i$, $i=1,\dots,M$,\footnote{Here we assume that
$S_i >0$, $i=1,\dots,M$.} 
then condition~\eqref{miso_priority_constraint} can be written as an equality
constraint:
\be
\frac{f_{1,1}}{S_1} = \dots =
\frac{f_{i,1}}{S_i} = \dots =
\frac{f_{M,1}}{S_M},
\label{miso_demand_proportional}
\ee
and the optimization
problem~\eqref{miso_objective}-\eqref{miso_priority_constraint}
turns into a \emph{linear program} (LP).

\begin{remark}
\label{remark:demandpriorities}
When selecting priorities $p_i$, they should be chosen such that they are \emph{not} functions of the demands $S_i$.
Otherwise, it can be shown that the resulting flows can violate the invariance principle (requirement 7 in Section \ref{subsec_node_review}) \citep{tampere11}.
This is true for all priorities considered in this paper.
For more detail on this, see the discussion in Section 2.1.3 of \citet{tampere11} and the references therein.
\end{remark}

\begin{remark}
Note that constraint~\eqref{miso_demand_proportional} cannot be trivially extended
to the MIMO node by replacing subindex $1$,
where it denotes the output, with subindex $j$.
Doing that, as is evident from~\cite{bliemer07}, indeed leads to a convenient
optimization problem formulation,
but sacrifices the flow maximization objective of the
node model, as was pointed out in~\citet{tampere11},
reducing the feasibility set more than necessary.
\end{remark}

For arbitrary priorities 
with $M=2$, condition~\eqref{miso_priority_constraint} becomes:
\begin{eqnarray}
& & f_{1,1} \leq
\max\left\{\frac{p_1}{p_1+p_2}R_1,\; R_1- S_2\right\};
\label{miso_2in_priority_constraint_1} \\
& & f_{2,1} \leq
\max\left\{\frac{p_2}{p_1+p_2}R_1,\; R_1- S_1\right\}.
\label{miso_2in_priority_constraint_2}
\end{eqnarray}
To give a hint how more complicated
constraint~\eqref{miso_priority_constraint} becomes
as $M$ increases,
let us write it out for $M=3$:
\be
f_{1,1} \leq
\max\left\{
\frac{p_1}{\sum_{i=1}^3 p_i}R_1,
\frac{p_1}{p_1+p_2}\left(R_1-S_3\right), 
\frac{p_1}{p_1+p_3}\left(R_1-S_2\right), 
R_1-\sum_{i=2,3} S_i
\right\}; \label{miso_3in_priority_constraint_1} \\
\ee
\be
f_{2,1} \leq
\max\left\{
\frac{p_2}{\sum_{i=1}^3 p_i}R_1,
\frac{p_2}{p_1+p_2}\left(R_1- S_3\right), 
\frac{p_2}{p_2+p_3}\left(R_1- S_1\right), 
R_1-\sum_{i=1,3} S_i
\right\}; \label{miso_3in_priority_constraint_2}
\ee
\be
f_{3,1} \leq
\max\left\{
\frac{p_3}{\sum_{i=1}^3 p_i}R_1,
\frac{p_3}{p_1+p_3}\left(R_1- S_2\right), 
\frac{p_3}{p_2+p_3}\left(R_1- S_1\right), 
R_1-\sum_{i=1,2} S_i
\right\}.\label{miso_3in_priority_constraint_3}
\ee
As we can see, right hand sides of 
inequalities~\eqref{miso_3in_priority_constraint_1}-\eqref{miso_3in_priority_constraint_3}
contain known quantities, and so for arbitrary priorities,
problem~\eqref{miso_objective}-\eqref{miso_priority_constraint} is also an LP.
For general $M$, however, building constraint~\eqref{miso_priority_constraint}
requires a somewhat involved algorithm.
Instead, we develop an algorithm for computing input-output flows
$f_{i1}^c$ that solves the maximization
problem~\eqref{miso_objective}-\eqref{miso_priority_constraint}.
For the special case of the MISO node, our solution algorithm is similar to the
solution algorithm of \citet{tampere11} (again, with a generalization of priorities
to be arbitrary quantities rather than capacities), so in the interest of brevity we
defer its statement to Appendix \ref{app_merge}.

The following theorem summarizes this section.
The proof will follow in Section~\ref{subsec_mimo_relaxed}, where we
prove a more general statement that has this problem as a special case.
\begin{theo}
The input-output flow computation algorithm of Appendix \ref{app_merge} constructs the
unique solution of the maximization
problem~\eqref{miso_objective}-\eqref{miso_priority_constraint}.
\label{theo_miso_optimal}
\end{theo}

\subsection{Single-Input-Multiple-Output (SIMO) Node: Relaxing the FIFO Condition}
\label{subsec_simo}
Recall our discussion of Figures~\ref{fig-simo} and~\ref{fig:networks}, and
the unrealistic spillback problems created by the FIFO condition.
A more realistic model would allow to specify how output
flows at any given junction can affect each other.
In the context of our example, we could say that 
congestion in off-ramp~1 affects mainline flow in lane~1,
congestion in off-ramp~3 affects mainline flow in lanes~4 and~5
(see striped areas in Figure~\ref{fig-simo}), while
traffic flow in lanes~2 and~3 of the mainline output link is
not affected by the traffic state in output links~1 and~3.
Additionally, we can say that congestion in the mainline output~2
affects flows directed to both off-ramps in full, 
while the traffic state in each of the off-ramps does not influence
the flow directed to the other off-ramp.

For purposes that will become clear in a moment, we refer to this partial blocking in terms of intervals of the lane group (e.g., the leftmost fifth of the lane group), rather than specific lanes (e.g., lane 1).
For example, referring again to Figure \ref{fig-simo}, off-ramp 1 becoming congested would affect the leftmost fifth of the mainline flow, and off-ramp 3 becoming congested would affect the rightmost two-fifths of the mainline flow.
Further, if we parameterize the five lanes serving the mainline flow as a unit interval, $[0,1]$, then we could say that off-ramp 1 becoming congested affects the interval $[0, \frac{1}{5}]$, and off-ramp 3 becoming congested affects the interval $[\frac{3}{5}, 1]$.

To formally describe this,
we introduce \emph{mutual restriction intervals}
$\etab_{j',j}= [y,z]\subseteq[0,1]$.
These parameters can be interpreted as follows:
\begin{itemize}
\item $\etab_{j',j}=[0,1]$ --- congestion in the output link $j'$
affects flow directed to the output link $j$ in full.
This is equivalent to the FIFO condition.
Obviously, $\etab_{j,j}\equiv[0,1]$.
\item $\etab_{j',j}=[0,0]$ (or any other interval of zero length) --- traffic
state in the output link $j'$
does not influence the flow directed to the output link $j$.
This is equivalent to no FIFO restriction.
\item $\etab_{j',j}=[y,z]\subset [0,1]$ --- traffic state in the output link
$j'$ affects a $|\etab_{j',j}|=z-y$ portion of the flow directed to the output
link $j$.
Moreover, we specify this influence as an interval, not just a scalar,
to capture the summary effect of multiple output links that may restrict
flow to the output link $j$ --- how, will be explained shortly.
\end{itemize}

Recall again the node in Figure~\ref{fig-simo}.
Here, the mutual restriction intervals, in matrix form, may be written as
\begin{equation*}
\etab = \left(\begin{array}{ccc}
[0, 1] & \left[0, \frac{1}{5}\right] & [0, 0] \\

[0, 1] & [0, 1] & [0, 1] \\

[0, 0] & \left[\frac{3}{5}, 1\right] & [0, 1]
\end{array}\right),
\end{equation*}
where $\etab_{j',j}$ is the $j',j$ element of the matrix.
The diagonal ($j=j'$) elements of the matrix, $\etab_{j,j}$, are always $[0,1]$ because they indicate that a movement becomes fully restricted when its destination link is blocked.
Note that the matrix is not symmetric in general, since a queue in link $j$ will typically not affect the same portion $[y,z]$ of $j'$-serving lanes as a queue in link $j'$ will affect the $j$-serving lanes (for example, $\etab_{2,1}= [0,1]$ in the above matrix because a queue in all five freeway lanes implies a queue in the leftmost freeway lane that serves offramp 1).

We have established how we encode the spatial (i.e., portion of lanes) extent of restriction into the restriction intervals. Now, we discuss how we use the restriction intervals to compute flows.
Recalling our example, suppose $\frac{R_1}{S_{1,1}}<\frac{R_3}{S_{1,3}}<1\leq\frac{R_2}{S_{1,2}}$.
In other words, demand for output links~1 and~3 exceeds the available supply
with output link~1 being more restrictive, while the demand directed
to the output link~2 can be satisfied.
Since output links~1 and~3 do not affect each other, we get
\begin{equation*}
f_{1,1} = R_1 \;\;\; \mbox{ and } \;\;\; f_{1,3} = R_3.
\end{equation*}
Flow $f_{1,2}$ is partially restricted by both output links~1 and~3:
\begin{equation*}
f_{1,2} = \left(1 - |\etab_{1,2}| - \left| \etab_{3,2} \setminus (\etab_{3,2} \cap \etab_{1,2}) \right|\right)S_{1,2} +
\frac{R_1}{S_{1,1}}|\etab_{1,2}|S_{1,2} +
\frac{R_3}{S_{1,3}}|\etab_{3,2} \setminus (\etab_{3,2} \cap \etab_{1,2})|S_{1,2},
\end{equation*}
where $|\etab_{1,2}|$ and $|\etab_{3,2}\setminus (\etab_{3,2} \cap \etab_{1,2})|$ denote the lengths of intervals
$\etab_{1,2}$ and $\etab_{3,2}\setminus (\etab_{3,2} \cap \etab_{1,2})$, respectively.
In the second interval, the intersection of $\etab_{3,2}$ and $\etab_{1,2}$ is removed because, as output link 1 is more restrictive by assumption, 
the output 3 will restrict $f_{1,2}$ only when output 1 is already restricting $f_{1,2}$ (of course, in this example, the intersection is the empty set, but as we will see momentarily, that is not always the case).
In this expression for $f_{1,2}$, the first term represents the unrestricted portion of flow (lanes~2 and~3);
the second term represents the portion of flow restricted by the output~1
(lane~1); and
the third term represents the portion of flow restricted by the output~3
(lanes~4 and~5).
The expression for $f_{1,2}$ can be rewritten as
\begin{equation}
f_{1,2} = S_{1,2} - 
\left(1 - \frac{R_1}{S_{1,1}}\right)|\etab_{1,2}|S_{1,2} -
\left(1 - \frac{R_3}{S_{1,3}}\right)|\etab_{3,2} \setminus (\etab_{3,2} \cap \etab_{1,2})|S_{1,2}. \label{eq:relaxed_fifo_flow_example}
\end{equation}
We believe \eqref{eq:relaxed_fifo_flow_example} is an intuitive representation of how the mutual restriction intervals affect flows.
The first term, $S_{1,2}$, is the maximum possible flow (i.e., the demand).
The second and third terms represent the portion of demand that cannot be fulfilled due to the restriction intervals.
For both of them, $(1- R_j / S_{1,j})$ represents the portion of flow that is affected by output link $j$ becoming congested, and $|\cdot|$, the length of the mutual restriction interval, is the degree to which this flow portion is affected by relaxed, or partial, FIFO.

Further, the computation in \eqref{eq:relaxed_fifo_flow_example} can be intuitively represented in the two-dimensional graphics of Figure \ref{fig-f12}.
In both subfigures, the area of the striped region is the flow $f_{1,2}$.
The entire rectangle, consisting of the region $[0, S_{1,2}]$ on the horizontal axis and ${0,1}$ on the vertical axis, has the area $S_{1,2}$ and represents the maximum possible flow.
The shaded regions correspond to the flow reductions in \eqref{eq:relaxed_fifo_flow_example}.

Focusing first on Figure \ref{fig-f12}(a), the top shaded region represents the flow reduction caused by congestion in output link 1.
One can see that its extent on the vertical axis is $\etab_{1,2} = [0, \frac{1}{5}]$, and its extent on the horizontal axis is $(1 - R_1 / S_{1,1}) S_{1,2} = (f_{1,1}/S_{1,1})S_{1,2}$.
The area of this shape, then, is $(1 - R_1 / S_{1,1}) |\etab_{1,2}| S_{1,2}$, which appears in \eqref{eq:relaxed_fifo_flow_example}.
Computing \eqref{eq:relaxed_fifo_flow_example}, then, is equivalent to measuring the shaded area in Figure \ref{fig-f12}(a).


\begin{figure}[htb]
\centering
\includegraphics[width=6in]{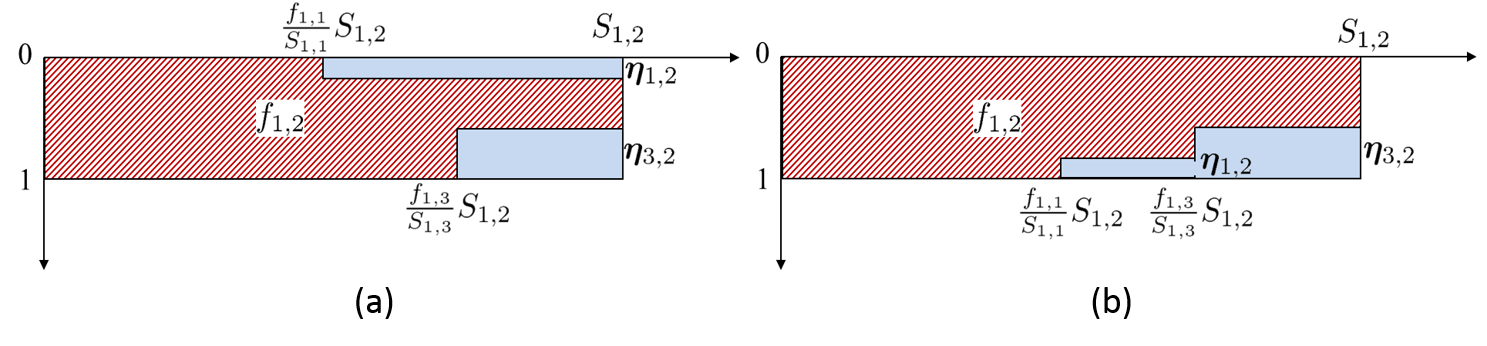}
\caption{Computing $f_{1,2}$ for the 1-input-3-output node example
with $\frac{R_1}{S_{1,1}}<\frac{R_3}{S_{1,3}}<1\leq\frac{R_2}{S_{1,2}}$,
$\etab_{13}=\etab_{31}=[0,0]$, $\etab_{32}=\left[\frac{3}{5}, 1\right]$
and two cases for $\etab_{1,2}$:  
(a) $\etab_{1,2}=\left[0, \frac{1}{5}\right]$;
(b) $\etab_{1,2}=\left[\frac{4}{5}, 1\right]$.}
\label{fig-f12}
\end{figure}

Now, considering both Figures \ref{fig-f12}(a) and \ref{fig-f12}(b), we see why we subtracted the intersection in the third term of \eqref{eq:relaxed_fifo_flow_example}.
In Figure \ref{fig-f12}(b), we have $\etab_{1,2} = \left[\frac{4}{5}, 1\right]$.
Here, restrictions from outputs~1 and~3 imposed on flow $f_{1,2}$ would
overlap in their affected lanes.
The intuition is that if traffic in lane~5 of the mainline output is
already restricted by the output~1, then the output~3 cannot do anything
more to restrict the flow in lane~5, it can only restrict flow in lane~4
of the mainline output.
Since the intersection in Figure \ref{fig-f12}(a)'s problem is the empty set, \eqref{eq:relaxed_fifo_flow_example} gives the area of its shaded area as well; \eqref{eq:relaxed_fifo_flow_example} is the proper equation for both cases.

We now generalize from two restriction intervals to an arbitrary number.
Define:
\begin{equation}
\QQ_{j',j} = \etab_{j',j}\times \left[\frac{f_{1,j'}}{S_{1,j'}}S_{1,j}, S_{1,j}\right],
\label{eq-rectangle-q-simo}
\end{equation}
where `$\times$' denotes a Cartesian product.
The Cartesian product of the two intervals gives us
a rectangle $\QQ_{j',j}$.
For example, in Figure~\ref{fig-f12}, $\QQ_{1,2}$ and $\QQ_{3,2}$ represent the grayed-out
rectangles.

Denote $\CA(\cdot)$ as the area of a two-dimensional shape.
For example, the expression for flow $f_{1,2}$
in our example, \eqref{eq:relaxed_fifo_flow_example}, can be replaced with a single, more general formula:
\begin{equation*}
f_{1,2} = S_{1,2} - \CA\left(\QQ_{1,2}\cup\QQ_{3,2}\right).
\end{equation*}

Now, we are ready to formulate the optimization problem for the 
general SIMO node with $N$ output links and $C$ commodities:
\be
\max\left(\sum_{j=1}^N\sum_{c=1}^C f_{1,j}^c\right),
\label{simo_objective}
\ee
subject to:
\begin{eqnarray}
& & f_{1,j}^c \geq 0, \;\; j=1,\dots,N, \; c=1,\dots,C \;
\mbox{ --- non-negativity constraint};
\label{simo_nonnegativity_constraint} \\
& & f_{1,j}^c \leq S_{1,j}^c, \;\; j=1,\dots,N, \; 
c=1,\dots,C \; \mbox{ --- demand constraint};
\label{simo_demand_constraint} \\
& & f_{1,j} \leq R_j, \;\; j=1,\dots,N \;
\mbox{ --- supply constraint};
\label{simo_supply_constraint} \\
& & \frac{f_{i,j}^c}{f_{1,j}} =
\frac{S_{1,j}^c}{S_{1,j}}, \;\;
j=1,\dots,N, \; c=1,\dots,C \;\; \mbox{ --- proportionality constraint for}\nonumber \\
& & \mbox{commodity flows};
\label{simo_proportionality_constraint}\\
& & f_{1,j} \leq S_{1,j} - \CA\left(\bigcup_{j'\neq j}\QQ_{j',j}\right), \;\;
j=1, \dots, N \;
\mbox{ --- relaxed FIFO constraint}.
\label{simo_rfifo_constraint}
\end{eqnarray}
For SIMO nodes with full FIFO, constraint~\eqref{simo_rfifo_constraint}
together with the supply constraint~\eqref{simo_supply_constraint}
translates into
\begin{equation}
f_{1,j} \leq S_{1,j} - \left(1 - \frac{f_{1,j^\ast}}{S_{1,j^\ast}}\right)S_{1,j} =
\frac{f_{1,j^\ast}}{S_{1,j^\ast}}S_{1,j},
\label{simo_fifo_constraint}
\end{equation}
where
\be
j^\ast = \arg\min_j\frac{R_{j^\ast}}{S_{1,j^\ast}},
\label{simo_restrictive_output}
\ee
and, since we are solving the flow maximization problem,
\eqref{simo_fifo_constraint} can be replaced with the equality constraint:
\begin{equation}
f_{1,j} = \frac{f_{1,j^\ast}}{S_{1,j^\ast}}S_{1,j}.
\label{simo_fifo_constraint_eq}
\end{equation}

For SIMO nodes with no FIFO, $\CA\left(\bigcup_{j'\neq j}\QQ_{j',j}\right)=0$,
which simplifies~\eqref{simo_rfifo_constraint} to the demand constraint,
and thus, constraint~\eqref{simo_rfifo_constraint} can be omitted in that case.

Next, we present the algorithm for solving the flow
maximization problem~\eqref{simo_objective}-\eqref{simo_rfifo_constraint}.
\begin{enumerate}
\item Initialize:
\begin{eqnarray*}
\St_{1,j}^c(0) & := & S_{1,j}^c; \\
\St_{1,j}(0) & := & S_{1,j};\\
\tetab_j(0) & = & [0, 0]; \\
V(0) & := & \{1, \dots, N\}; \\
k & := & 0; \\
& & j=1, \dots, N, \;\;\; c = 1, \dots, C.
\end{eqnarray*}
As before, $k$ denotes the iteration index;
$\St_{1,j}^c(k)$ ($\St_{1,j}(k)$) is the oriented demand per commodity (total accross commodities)
at iteration $k$;
$\tetab_j(k)$ is the portion of output flow to $j$ affected by the restricted
supply of other output links at iteration $k$,
and is the union of intervals that have become active as of iteration $k$;
and $V(k)$ is the set of output links still to be processed at iteration $k$.
\item If $V(k) = \emptyset$, stop.

\item For all output links $j\in V(k)$, find flow reduction factors:
\be
\alpha_j(k) = \frac{R_j}{\St_{1,j}(k)},
\label{simo-reduction-factors}
\ee
and find the most restrictive output link out of the remaining ones:
\be
j^\ast(k) = \arg\min_{j\in V(k)} \alpha_j(k).
\label{simo-most-restrictive}
\ee
\begin{itemize}
\item If $\alpha_{j^\ast(k)}(k) \geq 1$, then assign:
\begin{eqnarray}
f_{1,j}^c & = & \St_{1,j}^c(k), \;\; \forall j\in V(k), \;\; c=1,\dots, C; 
\label{simo-freeflow}\\
V(k+1) & = & \emptyset. \nonumber
\end{eqnarray}
\item Else, assign:
\begin{eqnarray}
\St_{1,j^\ast}^c(k+1) & = & \alpha_{j^\ast}(k)\St_{1,j^\ast}^c(k), \;\; c=1,\dots,C; 
\label{simo-step3-1} \\
\St_{1,j}^c(k+1) & = & \frac{\St_{1,j}(k+1)}{S_{1,j}}S_{1,j}^c, \;\; 
j\in V(k)\setminus\{j^\ast\}, \;\; c=1,\dots, C, \label{simo-step3-2} \\
& & \mbox{ where } \nonumber \\
\St_{1,j}(k+1) & = & \St_{1,j}(k) - 
S_{1,j}\left(|\etab_{j^\ast, j}| - |\tetab_j(k)\cap\etab_{j^\ast, j}|\right)
\left( 1 - \frac{\sum_{c=1}^C\St_{1,j^\ast}^c(k+1)}{S_{1,j^\ast}} \right),
\label{simo-step3-3} \\
\tetab_j(k+1) & = & \tetab_j(k) \cup \etab_{j^\ast, j}, \;\;
j\in V(k); \nonumber \\
f_{1,j}^c & = & \St_{1,j}^c(k+1), \;\;
j:~\tetab_j(k+1)=[0,1], \;\;
c=1,\dots, C;
\label{simo-step3-4} \\
V(k+1) & = & V(k)\setminus\{j:~\tetab_j(k+1)=[0,1]\}, \nonumber
\end{eqnarray}
where $|\tetab_j(k)\cap\etab_{j^\ast, j}|$ denotes the measure of the
interval intersection.\footnote{This set may be disjoint.}
\end{itemize}

\item Set $k := k+1$ and return to step~2.
\end{enumerate}
This algorithm takes no more than $N$ iterations to complete.

The algorithm can be understood as follows.
$V(k)$ is the set of output links who have not had their incoming flows determined as of iteration $k$.
In \eqref{simo-reduction-factors} and \eqref{simo-most-restrictive}, the output link whose remaining supply is most demanded (that is, whose ratio of remaining supply to total incoming demand is lowest) is identified and labeled $j^\ast(k)$.
If this most-demanded link is actually able to handle all of its demand, then we reach \eqref{simo-freeflow} and all node flows are in freeflow.
Otherwise, we reach \eqref{simo-step3-1}-\eqref{simo-step3-4}.
Equation \eqref{simo-step3-1} assigns the flow to the most-restricted output link $j^\ast(k)$ (this flow will be $R_{j^\ast}$, which can be seen by examining \eqref{simo-reduction-factors}).
Equations \eqref{simo-step3-2} and \eqref{simo-step3-3} perform the application of the mutual restriction interval for the just-filled link $j^\ast$ to all other output links $j$ in a two-step process.
The quantity $\St^c_{1,j}(k+1)$ is a ``running demand'' for commodity $c$ to link $j$ after relaxed FIFO has been applied to the original demand $S_{1,j}^c$.

This ``running demand'' that takes the activation of mutual restriction intervals is calculated in \eqref{simo-step3-3}.
Note that this calculation is exactly the calculation of a shaded area in the diagrams of Figure \ref{fig-f12}. In particular, $|\etab_{j^\ast, j}| - |\tetab_j(k)\cap\etab_{j^\ast, j}|$ defines the height of the union of a just-activated mutual restriction interval $\etab_{j^\ast, j}$ with the current active mutual restriction interval(s) $\tetab_j(k)$, and $1 - \left(\sum_{c=1}^C\St_{1,j^\ast}^c(k+1)\right)/\left(S_{1,j^\ast}\right)$ defines the width of the new rectangle.
Equation \eqref{simo-step3-2} simply scales each commodity's demand by the same ``blockage'' represented by the rectangle.
Similar to $\St^c_{1,j}(k+1)$, $\tetab_j(k+1)$ is a ``running mutual restriction interval'' that represents the union of mutual restriction intervals acting on flow to $j$ up to this point.
When link $j$ has a running mutual restriction interval of $[0,1]$, it has been fully blocked, and so its flow is assigned by its running demand in \eqref{simo-step3-4}.
This construction of the heights and widths of the rectangles will also be used in the multi-input-multi-output case in the next Section.


Finally, we state the result of this Section as a theorem.
As with Theorem \ref{theo_miso_optimal}, the proof is deferred to Section \ref{subsec_mimo_relaxed}.
\begin{theo}
The SIMO input-output flow computation algorithm constructs the
unique solution of the maximization
problem~\eqref{simo_objective}-\eqref{simo_rfifo_constraint}.
\label{theo_simo_optimal}
\end{theo}

\subsection{Multiple-Input-Multiple-Output (MIMO) Node}\label{subsec_mimo_relaxed}
To formally state the flow maximization problem for the MIMO node,
we need to extend the MISO optimization
problem~\eqref{miso_objective}-\eqref{miso_priority_constraint}
and the SIMO optimization 
problem~\eqref{simo_objective}-\eqref{simo_rfifo_constraint}.
While the generalization of the objective~\eqref{miso_objective},
\eqref{simo_objective} and
constraints~\eqref{miso_nonnegativity_constraint}-\eqref{miso_proportionality_constraint},
\eqref{simo_nonnegativity_constraint}-\eqref{simo_proportionality_constraint}
is straightforward,
extending the priority constraint~\eqref{miso_priority_constraint}
and the relaxed FIFO constraint~\eqref{simo_rfifo_constraint} to
the MIMO case requires a little more work.

First, we introduce the concept of oriented priorities:
\be
p_{i,j} = p_i \frac{\sum_{c=1}^C S_{i,j}^c}{\sum_{c=1}^C S_i^c}, \;\;\;
i=1, \dots, M, \;\; j=1, \dots, N.
\label{eq_oriented_priorities0}
\ee
Recall our discussion about the physical meaning of priorities in Section~\ref{subsec_miso}.
Qualitatively, priorities mean that drivers from each link $i$ ``compete''
to claim supply at rate $p_i$. 
The oriented priority $p_{i,j}$ can be thought of as the rate at which drivers 
from a link $i$ claim supply from the \emph{particular} link $j$.
This means that the ability of vehicles from $i$ to claim $j$'s supply is proportional to the portion of $i$ vehicles that are actually trying to claim $j$'s supply.

Before we present the algorithm, we need to define two items.

\bd
If, for a given input link $i$, whose demand cannot be satisfied
($\sum_{j=1}^N f_{i,j} < S_i$), there exists
at least one output link $j^\ast$, such that: (1) $S_{i,j^\ast}>0$; and
(2) $p_{i', j^\ast}f_{i,j^\ast} \geq p_{i,j^\ast}f_{i',j^\ast}$ 
for any $i'\neq i$,
we say that such output $j^\ast$ is \emph{restricting}
for input $i$.
\label{def-restrictive-output}
\ed
We need this definition to formulate the priority constraint, and it is only
valid for $M>1$.
If the output $j^\ast$ is restricting for input links $i'$ and $i''$,
then, according to this definition,
$p_{i'',j^\ast} f_{i',j^\ast} = p_{i',j^\ast} f_{i'',j^\ast}$.

\bd \label{def-restricting-set}
For a given input link $i$, define the set of restricting output links (as defined by Definition \ref{def-restrictive-output}) for that $i$ as
\begin{equation*}
W_i \triangleq \Big\{ j^\ast: \;\; R_{j^\ast} > S_{i,j^\ast} >0, \;\;
\exists \, i' \neq i \; \textnormal{s.t.} \; p_{i',j^\ast}f_{i,j^\ast} \geq p_{i,j^\ast}f_{i',j^\ast}
\Big\}.
\end{equation*}
\ed

Now, we can formulate the flow maximization problem for the general
MIMO node:
\be
\max\left(\sum_{i=1}^M\sum_{j=1}^N\sum_{c=1}^C f_{i,j}^c\right),
\label{mimo_objective}
\ee
subject to:
\begin{eqnarray}
& & f_{i,j}^c \geq 0, \;\; i=1,\dots,M, \; j=1,\dots,N, \; c=1,\dots,C \;
\mbox{ --- non-negativity constraint};
\label{mimo_nonnegativity_constraint} \\
& & f_{i,j}^c \leq S_{i,j}^c, \;\; i=1,\dots,M, \; j=1,\dots,N, \; 
c=1,\dots,C \; \mbox{ --- demand constraint};
\label{mimo_demand_constraint} \\
& & \sum_{i=1}^M f_{i,j} \leq R_j, \;\; j=1,\dots,N \;
\mbox{ --- supply constraint};
\label{mimo_supply_constraint} \\
& & \frac{f_{i,j}^c}{f_{i,j}} =
\frac{S_{i,j}^c}{S_{i,j}}, \;\; i = 1,\dots,M,\;
j=1,\dots,N, c=1,\dots,C \;\; \mbox{ --- proportionality constraint}\nonumber \\
& & \mbox{for commodity flows};
\label{mimo_proportionality_constraint}\\
& & \left.\begin{array}{cl}
\mbox{(a)} &
\mbox{For each input link $i$ such that } \\
& \sum_{j=1}^N f_{i,j} <  S_i, \;\; W_i\neq\emptyset; \\
\mbox{(b)} &
\mbox{For each input link $i$ such that $W_i\neq\emptyset$,} \\
& f_{i,j} \geq \frac{p_{i,j}}{\sum_{i'=1}^M p_{i',j}} R_j,
\;\; \forall j\in W_i, \\
\end{array}
\right\} \mbox{ --- priority constraint};
\label{mimo_priority_constraint}\\
& & f_{i,j} \leq S_{i,j} - \CA\left(\bigcup_{j'\in W_i\setminus\{j\}}\QQ_{j',j}^i \right),
\;\; i=1,\dots,M, \; j=1,\dots,N \;\;
\mbox{ --- relaxed FIFO} \nonumber \\
& & \mbox{ constraint}.
\label{mimo_rfifo_constraint}
\end{eqnarray}
Constraint~\eqref{mimo_priority_constraint} generalizes the
MISO priority constraint~\eqref{miso_priority_constraint},
with $W_i$ being the set of restricting outputs for input link $i$
as defined by Definition \ref{def-restricting-set}.
For each output link $j$, $j=1,\dots,N$, 
flows $\sum_{c=1}^C f_{i,j}^c$ fall into two categories:
(1) restricted by \emph{this} output link; and
(2) not restricted by \emph{this} output link.
Condition~\eqref{mimo_priority_constraint}(a) states that 
if the flow from an input $i$ is supply-constrained, there
do exist output links $j$, such that flows $f_{i,j}$
are of category~1;
and for a given output $j$, input-output flows of category~1 are allocated
proportionally to their priorities.
For every output link $j$, $j=1,\dots,N$, 
condition~\eqref{mimo_priority_constraint}(b) says that these category 1 flows
may take ``leftover'' supply after the category 2 flows into $j$ have been resolved.
For a SIMO ($M=1$) node, there is no competition between input flows, and
constraint~\eqref{mimo_priority_constraint} is satisfied automatically.

Constraint~\eqref{mimo_rfifo_constraint} generalizes the SIMO
relaxed FIFO constraint~\eqref{simo_rfifo_constraint}.
Here,
\begin{equation}
\QQ_{j',j}^i = \etab_{j',j}^i \times \left[\frac{f_{i,j'}}{S_{i,j'}}S_{i,j}, S_{i,j}\right], \label{eq-rectangle-q-mimo}
\end{equation}
which is a generalization of~\eqref{eq-rectangle-q-simo}, and $\CA(\cdot)$ denotes the area
of a two-dimensional object.
For MIMO nodes with full FIFO, constraint~\eqref{mimo_rfifo_constraint}
together with the supply constraint~\eqref{mimo_supply_constraint}
translates to:
\be
f_{i,j}\leq\min_{j'\in W_i}\left\{\frac{f_{i,j'}}{S_{i,j'}}\right\}S_{i,j},
\label{mimo_fifo_constraint}
\ee
and, since we are solving the flow maximization problem,
\eqref{mimo_fifo_constraint} can be replaced with the equality constraint:
\be
f_{i,j}=\min_{j^\ast\in W_i}\frac{f_{i,j^\ast}}{S_{i,j^\ast}}S_{i,j}.
\label{mimo_fifo_constraint2}
\ee
For MIMO nodes with no FIFO and for MISO ($N=1$) nodes,
$\CA\left(\bigcup_{j'\in W_i\setminus\{j\}}\QQ_{j',j}^i\right) = 0$,
and thus, constraint~\eqref{mimo_rfifo_constraint} degenerates into the
demand constraint.

\begin{remark}
Note that in the case of multiple input links
mutual restriction intervals are to be
specified per input link.
This can be justified by the following example.
In the node representing a junction
with 2 input and 3 output links,
shown in Figure~\ref{fig-intersection},
consider the influence of the output link 5
on the output link 4.
If vehicles enter links 4 and 5 from link 1, then
it is reasonable to assume that once link 5 is jammed
and cannot accept any vehicles, there is no flow
from 1 to 4 either, since vehicles queueing for a U-turn into link 4 would use the same lanes as those vehicles queueing for a left turn into link 5. 
In other words, $\etab_{5,4}^1=[0, 1]$.
On the other hand, if vehicles arrive from link 2,
blockage of the output link 5 may hinder, but not 
necessarily prevent, traffic from flowing into the output
link 4, and so $\etab_{5,4}^2=[1-\epsilon, 1]$, with some
$0< \epsilon<1$ that depends on the portion of 2's of lanes that allow both a right-turn and through movement.
So, from now on we will write $\etab_{j',j}^i$
and $\QQ_{j',j}^i$with index $i$
identifying the input link.
Example 1 in Section \ref{sec_examples} discusses a situation like this in detail.
\end{remark}

\begin{figure}[htb]
\centering
\includegraphics[width=1.5in]{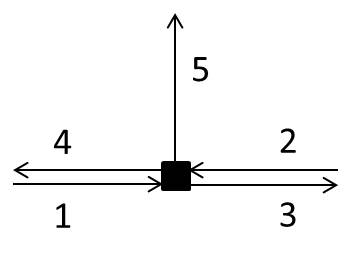}
\caption{Intersection node.}
\label{fig-intersection}
\end{figure}

Next, we present the algorithm that solves the flow maximization 
problem~\eqref{mimo_objective}-\eqref{mimo_rfifo_constraint}.
\begin{enumerate}
\item Initialize:
\begin{eqnarray*}
\Rt_j(0) & := & R; \\
U_j(0) & := & \left\{ i \in \{1,\dots,M\}: \; \sum_{c=1}^C S_{i,j}^c > 0\right\}; \\
\St_{i,j}^c(0) & := & S_{i,j}^c; \\
\St_{i,j}(0) & := & \sum_{c=1}^C \St_{i,j}^c(0);\\
\tetab_j^i(0) & := & [0, 0]; \\
k & := & 0; \\
& & i = 1,\dots,M, \;\;\; j = 1,\dots,N, \;\;\; c = 1,\dots,C.
\end{eqnarray*}
As before, $k$ denotes iteration;
$\Rt_j(k)$ is the remaining supply of output link $j$ at iteration $k$;
$U_j(k)$ is the set of input links $i$ that have nonzero demands towards $j$ (that is, at least one nonzero $S_{i,j}^c$) and have not had their flows assigned as of iteration $k$;
$\St_{1,j}^c(k)$ ($\St_{1,j}(k)$) is the oriented demand per commodity (total)
at iteration $k$;
$\tetab^i_j(k)$ is the portion of input-$i$-to-output-$j$ flow affected by the restricted
supply of other output links at iteration $k$
and is the union of intervals affecting that particular flow that have become active as of iteration $k$.

\item Define the set of output links that still need processing:
\[ V(k) = \left\{j: \; U_j(k) \neq \emptyset\right\}. \]
If $V(k) = \emptyset$, stop.

\item Check that at least one of the unprocessed input links has
nonzero priority, otherwise, assign equal positive priorities
to all the unprocessed input links:
\be
\pt_i(k) = \left\{\begin{array}{ll}
p_i, &
\mbox{ if there exists } i' \in \bigcup_{j\in V(k)} U_j(k):\;
p_{i'} > 0, \\
\frac{1}{\left|\bigcup_{j\in V(k)}U_j(k)\right|}, &
\mbox{ otherwise},
\end{array}\right.
\label{eq_nonzero_priorities}
\ee
where $\left|\bigcup_{j\in V(k)} U_j(k)\right|$ denotes the number
of elements in the union $\bigcup_{j\in V(k)} U_j(k)$;
and for each output link $j\in V(k)$ and input link $i\in U_j(k)$
compute oriented priority:
\be
\pt_{i,j}(k) = \pt_i(k) \frac{\sum_{c=1}^C S_{i,j}^c}{\sum_{c=1}^C S_i^c}.
\label{eq_oriented_priorities}
\ee

\item For each $j\in V(k)$, compute factors:
\be
a_j(k) = \frac{\Rt_j(k)}{\sum_{i\in U_j(k)}\pt_{i,j}(k)},
\label{eq_aj_factors}
\ee
and find the smallest of these factors:
\be
a_{j^\ast}(k) = \min_{j\in V(k)} a_j(k).
\label{eq_aj_min}
\ee
These factors $a_j(k)$ describe the ratio of the supply $R_j(k)$ to the demand placed upon $j$ by way of the oriented priorities $\pt_{i,j}(k)$\footnote{
If, as before, we consider oriented priorities as describing the ``rate'' at which supply is claimed, then $a_j(k) < a_{j'}(k)$ means that $R_j(k)$, which is being filled at rate $\sum_{i\in U_j(k)} \pt_{i,j}(k)$, is filled ``before'' $R_{j'}(k)$, which is being filled at rate $\sum_{i\in U_{j'}(k)} \pt_{i,j'}(k)$ (see \citet{wright_node_dynamic_2016} for more on this).}.
The link $j^\ast$ has the most demanded supply of
all output links.

\item Define the set of input links
whose demand does not exceed the allocated supply:
\[
\Ut(k) = \left\{i\in U_{j^\ast}(k):\;
\left( \sum_{j\in V(k)} \St_{i,j}(k) \right) \leq \pt_i(k)a_{j^\ast}(k)\right\}.
\]
\begin{itemize}
\item If $\Ut(k) \neq \emptyset$, then for all output links $j\in V(k)$ assign:
\begin{eqnarray}
f_{i,j}^c & = & \St_{i,j}^c(k), \;\;\; i\in\Ut(k),\;\; c=1,\dots,C; \label{eq_mimo_relaxed_freeflow} \\
\Rt_j(k+1) & = & \Rt_j(k) - \sum_{i\in\Ut(k)} f_{i,j}; \nonumber \\
U_j(k+1) & = & U_j(k) \setminus \Ut(k).  \nonumber
\end{eqnarray}
\item Else, for all input links $i\in U_{j^\ast}(k)$,
output links $j\in V(k)$ and commodities $c=1,\dots,C$, assign:
\begin{eqnarray}
\St_{i,j^\ast}^c(k+1) & = & \St_{i,j^\ast}^c(k)
\frac{\pt_{i,j^\ast}(k)a_{j^\ast}(k)}{\St_{i,j^\ast}(k)},
\;\;\; c=1,\dots,C;
\label{eq_mimo_relaxed_step5_1} \\
\St_{i,j}^c(k+1) & = & \St_{i,j}^c(k)\frac{\St_{i,j}(k+1)}{\St_{i,j}(k)},
\;\; i\in U_j(k)\cap U_{j^\ast}(k), \;\; j\in V(k)\setminus\{j^\ast\},
\label{eq_mimo_relaxed_step5_2} \\
\mbox{ where} & & \nonumber \\
\St_{i,j}(k+1) & = & \St_{i,j}(k) -
S_{i,j}\left(|\etab_{j^\ast ,j}^i| - |\tetab_j^i(k)\cap\etab_{j^\ast ,j}^i| \right)
\left( 1 - \frac{\sum_{c=1}^C \St_{i,j^\ast}^c(k+1)}{S_{i,j^\ast}} \right);
\label{eq_mimo_relaxed_step5_3} \\
\St_{i,j}^c(k+1) & = & \St_{i,j}^c(k), \;\;\;
i \not\in U_j(k)\cap U_{j^\ast}(k);\nonumber \\
\tetab_j^i(k+1) & = & \tetab_j^i(k) \cup \etab_{j^\ast, j}^i; \nonumber \\
f_{i,j}^c & = & \St_{i,j}^c(k+1), \nonumber \\
& & i\in U_j(k)\cap U_{j^\ast}(k), \; j\in V(k): \;
\tetab_j^i(k+1)=[0,1], \; c=1,\dots,C;
\label{eq_mimo_relaxed_step5_4} \\
\Rt_j(k+1) & = & \Rt_j(k) -
\sum_{i\in U_{j^\ast}(k):\etab_j^i(k+1)=[0,1]} \; f_{i,j}; \nonumber\\
U_j(k+1) & = & U_j(k) \setminus \left\{i\in U_{j^\ast}(k):\;
\tetab_j^i(k+1)=[0,1]\right\}. \nonumber
\end{eqnarray}
\end{itemize}

\item Set $k:=k+1$, and return to step 2.
\end{enumerate}

This algorithm takes no more than $M+N$ iterations to complete.

This algorithm can be best understood as a generalization of the algorithm for solving the SIMO node problem (Section \ref{subsec_simo}).
Similar to the SIMO problem, $V(k)$ is the set of output links whose input flows are not yet determined as of iteration $k$.
At each iteration $k$, $j^\ast \in V(k)$ is the output whose supply is most-demanded by its oriented priorities.
In step 5, it is determined whether any of the links that still want to send demand to $j^\ast$ (i.e., the $i \in U_{j^\ast}(k)$) will be able to send all their vehicles before this most-demanded link $j^\ast$'s supply is exhausted.
If so, then $\Ut(k)$ is nonempty, the first case of step 5 is entered, and these ``lucky'' links $\Ut(k)$ are able to satisfy all their demand in \eqref{eq_mimo_relaxed_freeflow}.

If not, then all links $i \in U_{j^\ast}(k)$ will still have leftover demand after $j^\ast$ is filled, and we enter the second case of step 5.
In \eqref{eq_mimo_relaxed_step5_1}, the demands into $j^\ast$ are scaled so that when their flows are assigned in \eqref{eq_mimo_relaxed_step5_4}, each $i \in U_{j^\ast}(k)$ will fill its priority-proportional share of $\Rt_{j^\ast}(k)$ (see the relevant calculation in the proof of Theorem \ref{theo_mimo_relaxed_optimal} to check this).
In \eqref{eq_mimo_relaxed_step5_2} and \eqref{eq_mimo_relaxed_step5_3}, relaxed FIFO is enforced on all flows into \emph{different} $j \neq j^\ast$ in accordance with their mutual restriction intervals $\etab_{j^\ast, j}^i$ (this is the same procedure as in \eqref{simo-step3-2} and \eqref{simo-step3-3} in the SIMO case).
Finally, \eqref{eq_mimo_relaxed_step5_4} sets all flows that can be found at this point, due to them being constrained by a FIFO constrant in the form of a running mutual restriction interval of $[0,1]$ (note that this includes the flows into $j^\ast$, as $\etab_{j^\ast, j^\ast}^i = [0,1] \; \forall i$ by definition).

The following lemma states that in the case of $N=1$,
the MIMO algorithm produces the same result as the MISO
algorithm described in Section \ref{subsec_miso}.
\begin{lemma}
The MISO algorithm is a special case of the MIMO algorithm with $N=1$.
\label{lemma_mimo}
\end{lemma}
\begin{proof}
The proof follows from the fact that for $N=1$,
formulae (\ref{eq_aj_factors})-(\ref{eq_aj_min}) result in
$a_{j^\ast}(k) = a_1(k)=\frac{\Rt_{1}(k)}{\sum_{i\in U_1(k)}\pt_i(k)}$ and $j^\ast = 1$.
\end{proof}

The following lemma states that in the case $M=1$, the MIMO algorithm
with relaxed FIFO condition produces the same result as the 
SIMO algorithm with relaxed FIFO condition described in
Section~\ref{subsec_simo}.
\begin{lemma}
The SIMO algorithm with the relaxed FIFO condition is a special case of the MIMO
algorithm with $M=1$.
\label{lemma_mimo_simo}
\end{lemma}
\begin{proof}
The proof follows from the fact that
for $M=1$, factors $a_j(k)$, defined in~\eqref{eq_aj_factors}, reduce to:
\[
a_j(k) = \frac{R_j \St_1(k)}{\St_{1,j}(k)},
\]
and
\begin{equation*}
j^\ast = \arg\min_{j\in V(k)}
\frac{R_j \St_1(k)}{\St_{1,j}(k)}
= \arg\min_{j\in V(k)}\frac{R_j}{\St_{1,j}(k)}. \qedhere
\end{equation*}
\end{proof}

The main result of this Section can be stated as the following theorem.
\begin{theo}
  Given a set of input links $i \in 1,\dots,M$, output links $j \in 1,\dots,N$,
  commodities $c \in 1,\dots,C$, priorities $\{p_i\}$, split ratios $\{\beta_{i,j}^c\}$, and
  mutual restriction intervals $\{ \etab_{j',j}^i \}$ ($j'\neq j$), the 
  algorithm of Section~\ref{subsec_mimo_relaxed} obtains the
  unique solution of the optimization 
  problem~\eqref{mimo_objective}-\eqref{mimo_rfifo_constraint}.
\label{theo_mimo_relaxed_optimal}
\end{theo}
\begin{proof}
  The priority constraint~\eqref{mimo_priority_constraint}
  makes this optimization problem non-convex, except in the
  special cases mentioned in Section~\ref{subsec_miso}.
  In fact, we conjecture that
  in the MIMO case, both with and without relaxation of the FIFO constraint,
  there is no way to verify a solution faster than re-solving the
  problem~\eqref{mimo_objective}-\eqref{mimo_rfifo_constraint}.
  Thus, we will prove optimality
  by showing that as our algorithm proceeds through iterations, it
  constructs the unique optimal solution.
   
  We may decompose the problem into finding the $M \cdot N \cdot C$ interrelated quantities
  $\{f_{i,j}^c\}$. The $C$ flows for each $(i,j)$ are further constrained by our
  commodity flow proportionality constraint~\eqref{mimo_proportionality_constraint};
  solving for one of $\{f_{i,j}^1,\dots,f_{i,j}^C\}$ also finds them all. Our
  task then becomes finding optimal values for each of $M \cdot N$ subsets
  $\{f_{i,j}^1,\dots,f_{i,j}^C\}$. Our algorithm finds at least one of these $M \cdot N$
  subsets per iteration $k$. The assignments
  are done by either equation~\eqref{eq_mimo_relaxed_freeflow} or
  equation~\eqref{eq_mimo_relaxed_step5_4}.
  Over iterations, subsets are assigned to build up the unique optimal solution.
  We can show that each subset assigned is optimal; that is, at least one of the
  constraints is tight.
  
  Consider \eqref{eq_mimo_relaxed_step5_3}-\eqref{eq_mimo_relaxed_step5_4},
  our implementation of the relaxed FIFO constraint. In step 4 of our algorithm,
  we identify a single output link as $j^\ast$. The minimization in this step picks
  out a single link as the most restrictive of all output links. By the relaxed
  FIFO construction, all $i\in U_{j^\ast}$ will feel a partial FIFO effect
  instigated by $j^\ast$ as the most restrictive link. For a generic
  $j \neq j^\ast$, equation~\eqref{eq_mimo_relaxed_step5_3} enforces the
  relaxed FIFO constraint by decaying the oriented demand $\St_{i,j}(k)$. In fact, this
  modified oriented
  demand acts as a proxy for the relaxed FIFO constraint. Since a $f_{i,j}^c$ that
  is restricted by partial FIFO will, by construction, never obtain $f_{i,j}^c=S_{i,j}^c$,
  the running quantity $\St_{i,j}^c(k)$ represents an ``effective'' oriented demand
  after application of partial FIFO. A flow that is not constrained by partial FIFO
  will nevertheless not exceed its demand, as $\St_{i,j}^c(0) = S_{i,j}^c$, and $\St_{i,j}^c(k)$
  only decreases across $k$ as relaxed FIFO constraints are considered
  (ensuring compliance with the demand constraint~\eqref{mimo_demand_constraint}).
  On the other hand, if it turns out that a flow is found
  $f_{i,j}^c=\St_{i,j}^c(k) \leq S_{i,j}^c$ for some $k$, then this flow value may have
  some constraint imposed by partial FIFO.

  Flows assigned by equation~\eqref{eq_mimo_relaxed_step5_4} for $j \neq j^\ast$ are
  constrained by a relaxed FIFO constraint~\eqref{mimo_rfifo_constraint}. To see this,
  note that
  \begin{equation}
  \CA( \QQ_{j',j}^i) = \left| \etab_{j',j}^i \right| \left( 1 - \frac{\sum_{c=1}^C f_{i,j'}^c}{S_{i,j'}} \right) S_{i,j}
  \end{equation}
  where, recall, $\QQ_{j',j}^i$ is defined by~\eqref{eq-rectangle-q-mimo}. As applied in
  \eqref{eq_mimo_relaxed_step5_3}, $j'=j^\ast$ always, and
  $f_{i,j^\ast}^c = \St_{i,j^\ast}^c$ by~\eqref{eq_mimo_relaxed_step5_4}.
  
  For a union of two rectangles $\QQ_{j',j}^i$ and $\QQ_{j'',j}^i$, we have
  \begin{equation}
	  \CA \left( \QQ_{j',j}^i \cup \QQ_{j'',j}^i \right) = \CA \left(\QQ_{j',j}^i \right)
	  + \CA \left( \QQ_{j'',j}^i \right) - \CA \left(\QQ_{j',j}^i \cap \QQ_{j'',j}^i \right)
	  \label{eq-mimo-union}
  \end{equation}
  and
  \begin{equation}
	  \CA \left( \QQ_{j',j}^i \cap \QQ_{j'',j}^i \right) = \left( \left| \etab_{j',j}^i \cap
	  \etab_{j'',j}^i \right| \right)
	  \min_{j^\sharp \in \{j',j''\}} \left( 1 - \frac{\sum_{c=1}^C f_{i,j^\sharp}^c}{S_{i,j^\sharp}} \right) S_{i,j}. \label{eq-mimo-jsharp}
  \end{equation}
  Note that $j^\sharp$ as defined in~\eqref{eq-mimo-jsharp} is the less-restricted of the
  two output links $j'$ and $j''$. This means that in the context of our algorithm, since
  the most-restricted link $j^\ast$ is picked at each iteration, subsequent links $j^\ast$
  at later iterations will always be less restricted than those in previous iterations.
  Equation~\eqref{eq_mimo_relaxed_step5_3} thus takes the union (with the subtraction
  of the intersection) as done in \eqref{eq-mimo-union}-\eqref{eq-mimo-jsharp}, and
  incorporates the relaxed FIFO constraint~\eqref{mimo_rfifo_constraint}.
  
  Now consider equations~\eqref{eq_mimo_relaxed_step5_1} and
  \eqref{eq_mimo_relaxed_step5_4}, which apply in the special case where $j=j^\ast$.
  We have
    \begin{align*}
      \sum_{c=1}^C \sum_{i \in U_{j^\ast}(k)} f_{i,j^\ast}^c &= \sum_{c=1}^C \sum_{i \in U_{j^\ast}(k)} \St_{i,j^\ast}(k) \frac{\pt_{i,j^\ast}(k) a_{j^\ast}(k)}{\sum_{c'=1}^C \St_{i,j^\ast}^{c'}(k)} \\
      &= \sum_{i \in U_{j^\ast}(k)} \pt_{i,j^\ast}(k) a_{j^\ast}(k) \\
      &= \left( \sum_{i \in U_{j^\ast}(k)} \pt_{i,j^\ast}(k) \right) \frac{\Rt_{j^\ast}(k)}{\sum_{i \in U_{j^\ast}(k)} \pt_{i,j^\ast}(k)} \\
      &= \Rt_{j^\ast}(k),
    \end{align*}
  so all the flows into link $j^\ast$ take up all available supply. These
  flows are thus constrained by the supply constraint,~\eqref{mimo_supply_constraint}.
\end{proof}

\begin{cor}
Theorems~\ref{theo_miso_optimal}
and~\ref{theo_simo_optimal} follow from
Theorem~\ref{theo_mimo_relaxed_optimal} as special cases.
\end{cor}

\begin{remark}
\label{remark:tampere}
Note that setting the priorities equal to input link capacities, $p_i = F_i$, and all
restriction intervals $\etab_{j',j}^i = [0,1]$, we recover the original node model and 
algorithm of \citet{tampere11}.
\end{remark}

\section{Examples}\label{sec_examples}
We present two examples to demonstrate the computation of node flows with our node model.
The first example is an extension of the example with four input links and four output links from \citet{tampere11}, where we have extended it to include the partial FIFO construction for some input links.

The second example represents an onramp to a freeway with a parallel managed lane facility as a node with three input links (onramp, freeway and managed lane) and two output links (freeway and managed lane).
In this example, we vary the input links' priorities to demonstrate how this affects the resulting node flows.

\subsection{Example One: An example with partial FIFO}
Consider the node schematically presented in Figure \ref{fig:4x4node}.
This node has four input links (links 1 through 4) and four output links (links 5 through 8).
\begin{figure}[htb]
\centering
\begin{subfigure}[b]{0.3\linewidth}
\centering
\includegraphics[width=1.5in]{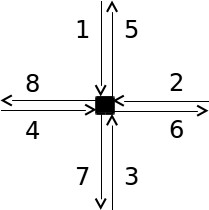}
\caption{\label{fig:4x4node}}
\end{subfigure}
\begin{subfigure}[b]{0.3\linewidth}
  \centering
  \includegraphics[width=50pt]{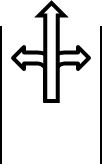}
  \caption{\label{fig:onelane}}
\end{subfigure}
\begin{subfigure}[b]{0.3\linewidth}
  \centering
  \includegraphics[width=100pt]{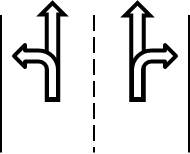}
  \caption{\label{fig:twolanes}}
\end{subfigure}

\caption{Items used in Example One: (\subref{fig:4x4node}) four-input, four-output node used, (\subref{fig:onelane}) one-lane structure and turn movements for input links 1 and 3 and (\subref{fig:twolanes}) two-lane structure and turn movements for input links 2 and 4.}
\label{fig-examplefifo}
\end{figure}

Suppose that our relevant parameters are:

\begin{tabular}{l l l l}
 $C=1$ & & & \\
 \hline
 $S_1=500$ & $S_2=2000$ & $S_3=800$ & $S_4=1700$ \\
 \hline
 $p_1=F_1=1000$ & $p_2=F_2=2000$ & $p_3=F_3=1000$ & $p_4=F_4=2000$ \\
 \hline
 $R_5=1000$ & $R_6=2000$ & $R_7=1000$ & $R_8=2000$ \\
 \hline
 $\beta_{1,5} = 0$ & $\beta_{1,6} = 0.1$ & $\beta_{1,7} = 0.3$ & $\beta_{1,8} = 0.6$ \\
 $\beta_{2,5} = 0.05$ & $\beta_{2,6} = 0$ & $\beta_{2,7} = 0.15$ & $\beta_{2,8} = 0.8$ \\
 $\beta_{3,5} = 0.125$ & $\beta_{3,6} = 0.125$ & $\beta_{3,7} = 0$ & $\beta_{3,8} = 0.75$ \\
 $\beta_{4,5} = \sfrac{1}{17}$ & $\beta_{4,6} = \sfrac{8}{17}$ & $\beta_{4,7} = \sfrac{8}{17}$ & $\beta_{4,8} = 0$ \\ \hline
\end{tabular}

where $F_i$ denotes the capacity of link $i$ and, we have omitted the commodity index $c$ on the demands and split ratios for readability, since there is only one commodity.

As stated above, these values are taken from the example of \citet{tampere11}.
Note that input links 2 and 4 have twice the capacity of input links 1 and 3.
Let us say that links 2 and 4 are two-lane roads, and links 1 and 3 are one-lane roads (see Figures \ref{fig:onelane} and \ref{fig:twolanes}).
Then, we can create a partial FIFO characterization of the multilane links 2 and 4.

%

\subsubsection{Defining restriction intervals}
\label{subsubsec_exampleeta}
Let us define restriction intervals
\begin{align*}
\etab_{j',j}^{2} &= \left\{ 
\begin{array}{ l c c c }
  \multicolumn{1}{r}{} & \multicolumn{1}{c}{j=5} & 
    \multicolumn{1}{c}{ j = 7} & \multicolumn{1}{c}{j = 8} \\
  \cline{2-4}
  j' = 5 & [0,1] & \emptyset & [\sfrac{1}{2}, 1] \\
  j' = 7 & \emptyset & [0,1] & [0, \sfrac{1}{2}] \\
  j' = 8 & [0,1] & [0,1] & [0,1] \\
\end{array} \right. \\
\etab_{j',j}^{4} &= \left\{ 
\begin{array}{ l c c c }
  \multicolumn{1}{r}{} & \multicolumn{1}{c}{j=5} & 
    \multicolumn{1}{c}{j = 6} & \multicolumn{1}{c}{ j = 7} \\
  \cline{2-4}
  j' = 5 & [0,1] & [0,\sfrac{1}{2}] & \emptyset \\
  j' = 6 & [0, 1] & [0,1] & [0,1] \\
  j' = 7 & \emptyset & [\sfrac{1}{2},1] & [0,1] \\
\end{array} \right.
\end{align*}
To read the above tables, recall that as written, $j'$ is the restricting link, and $j$ is the restricted link.
When $j'$ runs out of supply, the restriction intervals in that row become active.
Recall that the total restriction for a single movement $(i,j)$ is the union of all active intervals $\etab_{j',j}^i$.

We will describe the meaning of each element in the table for $\etab_{j',j}^2$ as part of this example.
This set of restriction intervals describe the behavior of partial FIFO blocking on input link 2 when an output link runs out of supply and becomes congested.
The intervals are meant to describe a two-lane road with left (onto link 7), straight (onto link 8), and right (onto link 5) movements.
The left turn movement is allowed on the left lane, the right turn movement is allowed on the right lane, and the straight movement is allowed on both lanes (see Fig \ref{fig:twolanes}).

The interval $\etab_{5,7}^2$ is the empty interval because, when the right turn movement is blocked (due to link 5 running out of supply), the right lane will begin to queue.
However, no vehicles intending to turn left (i.e., to link 7) will be in this lane, since the left turn movement is not allowed in the right lane.
Instead, the left-turning vehicles will all be in the left lane, which will not be blocked from link 7's spillback since no right-turning vehicles will be in the left lane (as the right turn movement is not allowed from the left lane).
Therefore, when the right turn movement is blocked, no left-turning vehicles will be in the blocked lanes, and the left turn movement (link 7) is unaffected by spillback in the right turn movement (link 5), or $\etab_{5,7}^2=\emptyset$.
This argument also applies to $\etab_{7,5}^2$, with the left and right lanes and turn movements switched.

The interval $\etab_{5,8}^2=[\sfrac{1}{2},1]$ encodes the effect of the blocking of the right lane described in the previous paragraph on the through movement.
The through movement is allowed in both lanes, so when the right lane is queueing from spillback from the right turn, the right lane will also have vehicles that are trying to go straight stuck in this queue.
Therefore, the right half-interval, or $[\sfrac{1}{2},1]$, of the through-movement-allowing-lanes will be blocked by queueing to take the right turn when link 5 spills back onto link 2.
This same argument explains how $\etab_{7,8}^2=[0,\sfrac{1}{2}]$, as the left half-interval will be blocked when the left turn spills back.
Note that since the total restriction interval for a movement is the union of all active intervals, this means that when both links 5 and 7 are congested and spill back onto link 2, the total restriction interval for movement $(2,8)$ is $\etab_{5,8}^2 \cup \etab_{7,8}^2 = [0,1]$, which means that the through movement becomes fully blocked when vehicles in link 2 are queueing to take both the left and right turn movements, and queueing in both the left and right lanes.

We have $\etab_{8,5}^2=[0,1]$ and $\etab_{8,7}^2=[0,1]$ to describe the effects of spillback from the through movement, link 8, on the turn movements.
Since the through movement is allowed in both lanes, vehicles trying to enter a blocked lane 8 will queue in both lanes, and all vehicles trying to take the turn movements will also be stuck in these queues.

Finally, the diagonal entries, $\etab_{5,5}^2, \etab_{7,7}^2,$ and $\etab_{8,8}^2$ are $[0,1]$ by definition.
This is because when $j'$ has no supply and is restricting, it is obviously also blocked by its own queue.

Input link 4 is assumed to have the same lane movement rules as link 2, so its restriction intervals are mirrored.
The one-lane input links, 1 and 3, are assumed to operate under full FIFO, as any queue would block the only lane.
That is, $\etab^i_{j',j}=[0,1]$, for all $j'$, $j$ and $i=1,3$.

\subsubsection{Solution}
We will now demonstrate the use of our node flow model solution algorithm to resolve the node flows from these supplies and demands.

First, we compute the oriented demands, with $S_{i,j} = \beta_{i,j} S_i$.

\begin{tabular}{l l l l}
 $S_{1,5}=0$ & $S_{1,6} = 50$ & $S_{1,7}=150$ & $S_{1,8} = 300$ \\ \hline
 $S_{2,5}=100$ & $S_{2,6} = 0$ & $S_{2,7}=300$ & $S_{2,8} = 1600$ \\ \hline
 $S_{3,5}=100$ & $S_{3,6} = 100$ & $S_{3,7}=0$ & $S_{3,8} = 600$ \\ \hline
 $S_{4,5}=100$ & $S_{4,6} = 800$ & $S_{4,7}=800$ & $S_{4,8} = 0.$ \\ \hline
\end{tabular}

We now outline step-by-step how our algorithm proceeds.
The numbers in each iteration $k$ correspond to the numbered steps in the algorithm in Section \ref{subsec_mimo_relaxed}.

\underline{$k=0:$}
\begin{enumerate}
	\item $U_5(0) = \{2,3,4\}; \quad U_6(0) = \{1,3,4\}; \quad U_7(0) = \{1,2,4\}; \quad U_8(0) = \{1,2,3\} $
	\item $ V(0) = \{5,6,7,8\}$
	\item No adjustments to the $\pt_{i}$ are made. The oriented priorities are: \\
		\begin{tabular}{l l l l}
		$\pt_{1,5} = 0$ & $\pt_{1,6} = 100$ & $\pt_{1,7}=300$ & $\pt_{1,8}=600$ \\
			$\pt_{2,5} = 100$ & $\pt_{2,6} = 0$ & $\pt_{2,7}=300$ & $\pt_{2,8}=1600$ \\
			$\pt_{3,5} = 125$ & $\pt_{3,6} = 125$ & $\pt_{3,7}=0$ & $\pt_{3,8}=750$ \\
			$\pt_{4,5} = 118$ & $\pt_{4,6} = 941$ & $\pt_{4,7}=941$ & $\pt_{4,8}=0$
		\end{tabular}
	\item $a_5(0) = 2.92; \quad a_6(0) = 1.72; \quad a_7(0) = 0.649; \quad a_8(0) = 0.678$ \\
		$a_{j^\ast}(0) = a_7(0) = 0.649$
	\item $\Ut(0) = \{1\}, \text{ as } S_1 = 500 \leq \pt_1(0) a_7(0) = 1000 \times 0.649$
		\begin{itemize}
			\item $\boldsymbol{f_{1,6}=S_{1,6}=50;} \quad \boldsymbol{f_{1,7}=S_{1,7}=150;} \quad \boldsymbol{f_{1,8}=S_{1,8}=300;}$
			\item $\Rt_6(1) = 2000-50=1950; \quad \Rt_7(1)=1000-150=850; \quad \Rt_8(1)=2000-300=1700$
			\item $U_5(1) = \{2,3,4\}; \quad U_6(1)=\{3,4\}; \quad U_7(1)=\{2,4\}; \quad U_8(1)=\{2,3\}$
		\end{itemize}
\end{enumerate}
In this iteration, we calculated that output link 7 was the most-demanded output link, and that of its demanding links, only input link 1 would be able to send its full demand before link 7 would run out of supply and become congested.
Input link 1 is able to send its full demand.

\underline{$k=1:$}
\begin{enumerate}
\setcounter{enumi}{1}
	\item $V(1) = \{5,6,7,8\}$
	\item No changes are made to priorities or oriented priorities.
	\item $a_5(1)=2.92; \quad a_6(1) = 1.83; \quad a_7(1) = 0.685; \quad a_8(0) = 0.723$ \\
		$a_{j^\ast}(1) = a_7(1) = 0.685$
	\item $\Ut(1) = \emptyset,$ as $\sum_{j \in V(k)} \St_{2j}(2) = 2000 \nleq \pt_2(1) a_7(1) = 1370,$ \\
	and $ \sum_{j \in V(k)} \St_{4j} =1700 \nleq \pt_4(1) a_7(1) = 1370,$
	\begin{itemize}
		\item $\boldsymbol{f_{2,7} = \pt_{2,7}(1)a_7(1) = 205.5 = \St_{2,7}(2)}$ \\
		  $\boldsymbol{f_{4,7} = \pt_{4,7}(1)a_7(1) = 644.5 = \St_{4,7}(2)}$
		\item $\tilde{\etab}^2_5 (2) = \tilde{\etab}^2_5(1) \cup \etab_{7,5}^2 = \emptyset \cup \emptyset = \emptyset$ \\
		  $\tilde{\etab}^2_8 (2) = \tilde{\etab}_8^2(1) \cup \etab_{7,8}^2 = \emptyset \cup [0, \sfrac{1}{2}] = [0, \sfrac{1}{2}]$ \\
		   $\tilde{\etab}^4_5 (2) = \tilde{\etab}^4_5(1) \cup \etab_{7,5}^4 = \emptyset \cup \emptyset = \emptyset$ \\
		 $\tilde{\etab}^4_6 (2) = \tilde{\etab}^4_6(1) \cup \etab_{7,6}^4 = \emptyset \cup [\sfrac{1}{2}, 1] = [\sfrac{1}{2}, 1]$
		\item $\St_{2,8}(2) = 1600 - 1600( \sfrac{1}{2}) \left( 1 - \frac{205.5}{300} \right) = 1348$ \\
		  $\St_{4,6}(2) = 800 - 800( \sfrac{1}{2}) \left( 1 - \frac{644.5}{800} \right) = 772.25$
		\item	$\Rt_7(2)=850-644.5-205.5=0$
		\item $U_5(2) = \{2,3,4\}; \quad U_6(2)=\{3,4\}; \quad U_7(2)=\emptyset; \quad U_8(2)=\{2,3\}$
	\end{itemize}
\end{enumerate}
In this iteration, we found that, after accounting for input link 1's flow in the previous iteration, output link 7 remains the most-demanded output link.
However, none of the input links will be able to send all of their demand before link 7's supply is exhausted.
Therefore, the remaining supply of link 7 is assigned priority-proportional to input links 2 and 4, and restriction intervals $\etab_{7j}^i$ become active.
The restriction intervals for links 2 and 4 reflect the overlapping lanes of the restricted movement $(i,j)$ with the restricted movement $(i,7)$.
Since $(i,7)$ is a turn movement for $i=2,4$, this means that the through movement has its overlapping half-interval restricted, and the other turn movement is not restricted.
The new values for $\St_{2,8}(2)$ and $\St_{4,6}(2)$ describe the new maximum flow for these movements, the subtracted term representing the vehicles that will be stuck in the queued lane when the congestion in the turn movement spills back.

Since input link 3 had no demand for the movement $(3,7)$, there are no vehicles that will queue for that movement when output link 7 becomes congested, and so there are no restriction intervals that become active for input link 3 at this time.

\underline{$k=2:$}
\begin{enumerate}
\setcounter{enumi}{1}
	\item $V(2) = \{5,6,8\}$
	\item No changes are made to priorities or oriented priorities.
	\item $a_5(2) = 2.82; \quad a_6(2) = 1.83; \quad a_8(2) = 0.723$ \\
		$a_{j^\ast}(2) = a_8(2) = 0.723$
	\item $\Ut(2) = \emptyset,$ as $\sum_{j \in V(k)} \St_{2j}(2) = 205.5 + 1348 = 1553.5 \nleq \pt_2(2) a_8(2) = 1446$, \\
	and $\sum_{j \in V(k)} \St_{3j}(2) = 800 \nleq \pt_3(2) a_8(2) = 723$
	\begin{itemize}
		\item $\boldsymbol{f_{2,8} = \pt_{2,8}(2)a_8(2) = 1157.4 = \St_{2,8}(3)}$ \\
		  $\boldsymbol{f_{3,8} = \pt_{3,8}(2)a_8(2) = 542.6 = \St_{3,8}(3)}$
		\item $\tilde{\etab}^2_5 (3) = \tilde{\etab}^2_5(2) \cup \etab_{8,5}^2 = \emptyset \cup [0,1] = [0,1]$ \\
		   $\tilde{\etab}^3_5 (3) = \tilde{\etab}^3_5(3) \cup \etab_{8,5}^3 = \emptyset \cup [0,1] = [0,1]$ \\
		 $\tilde{\etab}^3_6 (3) = \tilde{\etab}^3_6(2) \cup \etab_{8,6}^3 = \emptyset \cup [0,1] = [0, 1]$
		\item $\St_{2,5}(3) = 100 - 100( 1 ) \left( 1 - \frac{1157.4}{1600} \right) = 72.3$ \\
		  $\St_{3,5}(3) = 100 - 100( 1 ) \left( 1 - \frac{542.6}{800} \right) = 67.8$ \\
		  $\St_{3,6}(3) = 100 - 100( 1 ) \left( 1 - \frac{542.6}{800} \right) = 67.8$
	  \item $\boldsymbol{f_{2,5} = \St_{2,5}(3) = 72.3$} \\
	    $\boldsymbol{f_{3,5} = \St_{3,5}(3) = 67.8$} \\
	    $\boldsymbol{f_{3,6} = \St_{3,6}(3) = 67.8$}
		\item	$\Rt_8(3)=1700-1157.4-542.6=0$ \\
		  $\Rt_5(3) = 1000 - 72.3 - 67.8 = 859.9$ \\
      $\Rt_6(3) = 1950 - 67.8 = 1882.2$
		\item $U_5(3) = \{4\}; \quad U_6(3)=\{4\}; \quad U_8(3)=\emptyset$
	\end{itemize}
\end{enumerate}

This iteration is similar to the previous iteration in that an output link, this time link 8, has its supply exhausted.
However, the two input links under consideration, links 2 and 3, have a full FIFO restriction on their flows when link 8 runs out of supply.
For link 2, this is because $(2,8)$ is the through movement whose queue takes up both lanes as described in Section \ref{subsubsec_exampleeta}.
For link 3, this is because it has only one lane, so any queue blocks all traffic.
The remaining flows for input links 2 and 3 are therefore determined by this FIFO restriction.

\underline{$k=3:$}
\begin{enumerate}
\setcounter{enumi}{1}
	\item $V(3) = \{5,6\}$
	\item No changes are made to priorities or oriented priorities.
	\item $a_5(3) = 7.29; \quad a_6(3) = 2.35$ \\
		$a_{j^\ast}(3) = a_6(3) = 2.35$
	\item $\Ut(3) = \{4\},$ as $\sum_{j \in V(k)} \St_{4j}(3) = 100 + 772.25 = 872.25 \leq \pt_4(3) a_6(3) = 4705.5$,
	\begin{itemize}
		\item $\boldsymbol{f_{4,5} = \St_{4,5}(3) = 100}$ \\
		  $\boldsymbol{f_{4,6} = \St_{4,6}(3) = 772.25}$
		\item $\Rt_5(4) = 859.9 - 100 = 759.9$ \\
      $\Rt_6(4) = 1882.2 - 759.9 = 1122.3$
		\item $U_5(4) = \emptyset; \quad U_6(4)=\emptyset$
	\end{itemize}
\end{enumerate}

This iteration finds the remaining flows.
Note that link 4's through movement, (4,6), is already under effect of partial FIFO, so it is not able to fill its entire demand $S_{4,6}$.

\underline{$k=4:$}
\begin{enumerate}
\setcounter{enumi}{1}
	\item $V(4) = \emptyset$, so the algorithm terminates. All flows have been found.
\end{enumerate}

It is useful to compare the solution to this example problem to that of the example in \citet{tampere11} (Tables \ref{tab:tampere_new} and \ref{tab:tampere_old}).
As noted above, all values for this example are the same as those in \citet{tampere11}, except in the present example we have added the partial FIFO relaxation to some movements.
In \citet{tampere11}, with the full FIFO restriction, iteration 0 proceeded the same, but at iteration 1, when link 7's supply is exhausted, \emph{all} of links 2 and 4's flows were set due to the full FIFO restriction.
Then, the full-FIFO example finds that input link 3 is able to have all of its demand satisfied, as the high amount of spillback from link 7 to links 2 and 4 left more supply available for link 3 in link 8 than in the present example.

\begin{table}[h]
\centering
\begin{tabular}{l r r r r}
\hline
& \multicolumn{1}{c}{$j=5$} & \multicolumn{1}{c}{$j = 6$} & \multicolumn{1}{c}{$j = 7$} & \multicolumn{1}{c}{$j=8$} \\
\cline{2-5}
$i=1$ & 0 & 50 & 150 & 300 \\
$i=2$ & 72.3 & 0 & 205.5 & 1157.4 \\
$i=3$ & 67.8 & 67.8 & 0 & 542.6 \\
$i=4$ & 100 & 772.3 & 644.5 & 0 \\ \hline
\end{tabular}
\caption{Flows $f_{i,j}$ resulting from enforcing partial FIFO as specified in the present example}
\label{tab:tampere_new}
\end{table}

\begin{table}[h]
\centering
\begin{tabular}{l r r r r}
\hline
& \multicolumn{1}{c}{$j=5$} & \multicolumn{1}{c}{$j = 6$} & \multicolumn{1}{c}{$j = 7$} & \multicolumn{1}{c}{$j=8$} \\
\cline{2-5}
$i=1$ & 0 & 50 & 150 & 300 \\
$i=2$ & 68.5 & 0 & 205.5 & 1096 \\
$i=3$ & 100 & 100 & 0 & 600 \\
$i=4$ & 80.6 & 644.5 & 644.5 & 0 \\ \hline
\end{tabular}
\caption{Flows $f_{i,j}$ resulting from enforcing full FIFO (from \citet[Table 8]{tampere11})}
\label{tab:tampere_old}
\end{table}

In other words, if this example is meant to model a junction where two input links (2 and 4) have two lanes and twice the supply as the other two, a full FIFO construction may have led to the ``unrealistic spillback'' phenomenon we discussed in Section \ref{subsec_fifo_implications}, with queueing for a turning movement causing spillback into through movements and non-competing turning movements.
Further, as this example demonstrates, an overly aggressive spillback model can lead to knock-on effects in differences in other input links' flows, as shown by link 3 being able to send all its supply when the two-lane links suffer full FIFO from a turn movement queue, but is not able to when the two-lane links have only a partial FIFO restriction.

\subsection{Example Two: Variation of input links' priorities}
Consider the node schematically presented in Figure \ref{fig-examplehov}.
This node has three input links (links 1 through 3) and two output links (4 and 5).
This node might represent an onramp (link 3) joining a freeway (links 1 and 4) with a parallel managed lane facility (links 2 and 5).
In this situation, we will have two vehicle commodities, with $c=1$ representing vehicles that cannot enter the managed lane facility, and $c=2$ representing vehicles that may, but are not required to, enter the managed lane facility.

\begin{figure}[htb]
\centering
\includegraphics[width=1.5in]{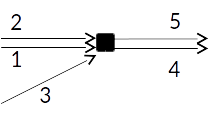}
\caption{Three-input, two-output node used in Example Two.}
\label{fig-examplehov}
\end{figure}

Let us define our parameters as follows:

\begin{tabular}{l l l l}
 $C=2$ & & & \\
 \hline
 $S_1^1=1700$ & $S_2^1=0$ & $S_3^1=400$ & \\
 $S_1^2=200$ & $S_2^2=500$ & $S_3^2=200$ &\\
 $F_1 = 4000$ & $F_2 = 2000$ & $F_3 = 1000$ &\\
 \hline
 $R_4=2000$ & $R_5=1000$ & &\\
 \hline
 $\beta_{1,4}^1 = 1$ & $\beta_{1,5}^1 = 0$ & $\beta_{1,4}^2 = 0.2$ & $\beta_{1,5}^2 = 0.8$ \\
 $\beta_{2,4}^1 = 1$ & $\beta_{2,5}^1 = 0$ & $\beta_{2,4}^2 = 0.1$ & $\beta_{2,5}^2 = 0.9$ \\
 $\beta_{3,4}^1 = 1$ & $\beta_{3,5}^2 = 0$ & $\beta_{3,4}^2 = 0.5$ & $\beta_{3,5}^2 = 0.5$ \\ \hline
 \multicolumn{4}{l}{$\etab_{j',j}^i = [0,1]$ for all $i,j,j'$ (Full FIFO).} \\ \hline
\end{tabular}.

To demonstrate the effect of varying input link priorities on node flows, we will use three different priority-assignment schemes with the above values.
This also demonstrates the usefulness of our node model in studying the effect of priorities on node flows: 
For the sake of brevity, we will not show intermediate calculations here as was done in the previous example, but simply show and discuss the final node flows.

\textbf{Capacity-proportional priorities:} \citet{tampere11} recommended assigning downstream supply among the node flows in proportion to the input links' capacities.
As we mentioned in Remark \ref{remark:tampere}, setting the priorities in our node model equal to capacities and using strict FIFO recovers the \citet{tampere11} node model as a special case.
So, setting $p_i=F_i$ for $i=1,2,3$ and calculating node flows, we obtain the flows shown in Table \ref{tab:capacityprop}.

\begin{table}[h]
\centering
\begin{tabular}{l l r r}
  \hline
  & & \multicolumn{1}{c}{$j=4$} & \multicolumn{1}{c}{$j=5$} \\
  \cline{3-4}
  $i=1$ & $c = 1$ & 1552.1 & 0 \\
  & $c=2$ & 36.52 & 146.1 \\
  $i=2$ & $c = 1$ & 0 & 0 \\
  & $c=2$ & 50 & 450 \\
  $i=3$ & $c = 1$ & 289.1 & 0 \\
  & $c=2$ & 72.28 & 72.28 \\
  \hline \hline
  \multicolumn{2}{l}{$R_j - \sum_i \sum_c f_{i,j}^c$} & 0 & 331.6 \\
  \hline
\end{tabular}
\caption{Flows $f_{i,j}^c$ from Example Two with capacity-proportional priorities}
\label{tab:capacityprop}
\end{table}

Here, the GP link ($j=4)$ had its supply filled, while the managed lane link ($j=5$) had some leftover supply.
The input GP link ($i=1$) and onramp link ($i=3$) both had some unfulfilled special ($c=2$) vehicles that could have taken this leftover supply in the managed lane link, but they were restricted by strict FIFO from link 1's supply exhausting.
Note that all of the demand of the input managed lane link ($i=2$) was able to be satisfied.
This will be in contrast to the next set of priorities:

\textbf{Demand-proportional priorities:} As we mentioned in Remark \ref{remark:demandpriorities} and elsewhere, in practice one should not assign priorities in proportion to the input links' demand, as this would produce flows that would violate nodal invariance principles \citep{tampere11}.
However, for the purpose of this example, we will use priorities equal to the total demands purely for demonstration purposes, so that the resulting flows can be contrasted with other priority examples.

Setting $p_1 = 1900, p_2 = 500, p_3 = 600$, and calculating node flows, we obtain the flows shown in Table \ref{tab:demandprop}.

\begin{table}[h]
\centering
\begin{tabular}{l l r r}
  \hline
  & & \multicolumn{1}{c}{$j=4$} & \multicolumn{1}{c}{$j=5$} \\
  \cline{3-4}
  \multirow{1}{*}{$i=1$} & $c = 1$ & 1484.7 & 0 \\
  & $c=2$ & 34.93 & 139.7 \\
  \multirow{1}{*}{$i=2$} & $c = 1$ & 0 & 0 \\
  & $c=2$ & 43.67 & 393.0 \\
  \multirow{1}{*}{$i=3$} & $c = 1$ & 349.3 & 0 \\
  & $c=2$ & 87.33 & 87.33 \\
  \hline \hline
  \multicolumn{2}{l}{$R_j - \sum_i \sum_c f_{i,j}^c$} & 0 & 379.9 \\
  \hline
\end{tabular}
\caption{Flows $f_{i,j}^c$ from Example Two with demand-proportional priorities}
\label{tab:demandprop}
\end{table}

Note that, while similar to the previous example in that link $j=4$ has its supply exhausted and link $j=5$ has some extra $c=2$ special vehicles that cannot reach it, this example is different from the previous example in that the input managed lane link ($i=2$) is not able to send all of its vehicles.
Vehicles exiting this link encounter a queue entering link $j=4$ when that link fills, which blocks vehicles entering link $j=5$ due to this example's strict FIFO construction.
This blockage leads to more spillback and less flow through the node, as evidenced by the higher amount of leftover supply in link $j=5$.

\textbf{Onramp-preference priorities:} Some authors (e.g., \citet{coogan2014freeway}, \citet{gomes06})  have suggested that the particular problem of a relatively small onramp merging with a larger freeway should be modeled such that vehicles from the onramp are able to claim any available downstream supply \emph{first}, so that a queue would always appear on the mainline before the onramp.
In addition, this type of node model at onramp junctions is implicit in many ramp metering control schemes, from classic schemes such as ALINEA \citep{alinea91} to recent neural-network-based reinforcement learning approaches (e.g., \citet{rezaee2014decentralized, belletti2017expert}).
In particular, a ramp metering control problem often defines the onramp \emph{flow} as an exogenous input that is controllable by metering, rather than the onramp \emph{demand} (that is, the freeway is assumed to always accept all demand that the meter sends from the onramp).
In our node model, we can create this special case by proper initialization of the priorities.
In particular, setting $p_1 = p_2 = 0$ and $p_3 > 0$, our node model will assign all necessary supply to the onramp on its first iteration.
Then, on subsequent iterations, in \eqref{eq_nonzero_priorities} input links $i=1,2$ will have nonzero priority and be able to claim the leftover supply.
Following this procedure, we obtain the flows shown in Table \ref{tab:actm}.

\begin{table}[h]
\centering
\begin{tabular}{l l r r}
  \hline
  & & \multicolumn{1}{c}{$j=4$} & \multicolumn{1}{c}{$j=5$} \\
  \cline{3-4}
  \multirow{1}{*}{$i=1$} & $c = 1$ & 1416.7 & 0 \\
  & $c=2$ & 33.33 & 133.3 \\
  \multirow{1}{*}{$i=2$} & $c = 1$ & 0 & 0 \\
  & $c=2$ & 50 & 450 \\
  \multirow{1}{*}{$i=3$} & $c = 1$ & 400 & 0 \\
  & $c=2$ & 100 & 100 \\
  \hline \hline
  \multicolumn{2}{l}{$R_j - \sum_i \sum_c f_{i,j}^c$} & 0 & 316.7 \\
  \hline
\end{tabular}
\caption{Flows $f_{i,j}^c$ from Example Two with onramp-preference priority scheme}
\label{tab:actm}
\end{table}

Note that we have indeed satisfied all demand of the onramp link ($i=3$) with this priority scheme.
In this particular example, we have $p_1 = p_2 = 0.5$ when we reset their priorities, as prescribed in \eqref{eq_nonzero_priorities}.
However, a different version of \eqref{eq_nonzero_priorities} that resets the priorities to nonequal values (e.g., $p_i = F_i$) is possible. 
In this case, assigning higher priority to input link $i=1$ may prevent $i=2$ from sending its full demand.
This further demonstrates the variability of node flows to the priority-assignment scheme.

In summary, Example Two shows that input priorities determine not only define the upstream directions of queue formations,
but can also affect the node throughput.

\section{Conclusion}\label{sec_conclusion}
This paper discussed and addressed several node-rooted modeling issues that the authors have encountered in
macroscopic simulation of large and/or high-dimensional road networks.
First, the node model framework of~\citet{tampere11} was extended to the case of arbitrary
input priorities, with discussion of what input priorities might represent.
Our node model also allows priorities to have zero values, and, in the style of the literature,
is formulated as a mathematical optimization problem.

Second, we discussed a tradeoff between unrealistic spillback behavior and a need to model a road with multiple links, and traced this tradeoff back to the first-in-first-out (FIFO) constraint that is present in common macroscopic node models.
To resolve this dilemma, we proposed a generalization of the FIFO constraint that is easily applicable to node models of the common~\citet{tampere11} framework.
In addition, our FIFO relaxation is intuitive due to the ability to illustrate its effects with simple geometrical shapes.

We believe our node model, and our relaxed FIFO construction, is widely applicable across link models.
Like any node model in the~\citet{tampere11} framework, it can be applied with any link model that defines supplies, demands, and some form of priority.
In addition, while this paper's discussion of the node flow problem strictly deals with discrete-time simulations (i.e., at some timestep $t$, some fixed supplies $S_{i,j}^c(t)$ and demands $R_j(t)$ are given, and the node model computes some set of throughflows $f_{i,j}^c$ at that timestep), which, like the classic CTM are usually of fixed temporal step size, we have written a companion paper~\citep{wright_node_dynamic_2016} where the flows computed by this paper's node model are shown to be equal to those computed by a particular hybrid dynamic system. 
In the dynamic system setting, the flows $f_{i,j}^c$ are given in terms of a time differential, $\frac{d}{dt} f_{i,j}^c(t)$.
Thus, this paper's node model (and all other node models of the~\citet{tampere11} framework, of which this paper's node model is a generalization) and its relaxed FIFO construction should be applicable in simulations with continuous-time link models (where the supplies and demands are given by differential equations) or variable-step-size simulations such as in~\citet{raadsen_efficient_2016}.

As a final note, we want to emphasize the usefulness of our generalization of the \citet{tampere11}-style ``general class of node models'' framework to multi-commodity flows.
Many other works in the node model literature (e.g., \citet{gibb_model_2011, corthout_non-unique_2012, smits_family_2015, jabari_node_2016}\footnote{
Note, however, that \citet{jabari_node_2016}, in order to specify the non-concurrence of conflicting turn movements, defined different movements as different commodities with their own fundamental diagrams.
}), but as noted, we considered multi-commodity flows, with several commodities in the same link and taking the same movement, in this paper.
It turns out that our expression of multi-commodity flows, while simple (Requirement 8 in Section \ref{subsec_node_review}), can be quite powerful.
In a forthcoming paper \citep{wright2017generic}, we make use of this construction, and give each commodity its own fundamental diagram, to extend this paper's node model to the second order of macroscopic traffic models.

All of our results have been presented in the form of constructive
computational algorithms that are readily implementable in
macroscopic traffic simulation software.

\section*{Acknowledgements}\label{sec_acknowledgement}
We would like to express great appreciation to our colleagues
Elena Dorogush and Ajith Muralidharan for sharing ideas,
Ramtin Pedarsani, Brian Phegley and Pravin Varaiya for their
critical reading and their help in clarifying some theoretical issues.

This research was funded by the California Department of Transportation.

\pagebreak
\appendix
\section{Notation}\label{app_notation}
\begin{longtable}[l]{l p{4.5in} l}
	Symbol & Definition \\ 
	\hline
	$k$ & Iteration index \\ 
	$i$ & Index of links entering a node \\ 
	$M$ & Number of links entering a node \\ 
	$j$ & Index of links exiting a node \\ 
	$N$ & Number of links exiting a node \\ 
	$c$ & Vehicle commodity index \\ 
	$C$ & Number of vehicle commodities \\ 
	$S_i^c$ & Demand for commodity $c$ of link $i$ \\ 
	$S_i$ & Total demand for input link $i$, $S_i = \sum_{c=1}^C S_i^c$ \\ 
	$R_j$ & Supply of link $j$ \\ 
	$p_i$ & Priority of input link $i$ \\ 
	$f_{i,j}^c$ & Flow of vehicle commodity $c$ from link $i$ to link $j$ \\ 
	$\etab_{j,j'}^i$ & Mutual restriction interval of link $j$ onto link $j'$ for link $i$ \\ 
	$U(k)$ & Set of input links whose flows have yet to be fully determined as of iteration $k$ \\ 
	$V(k)$ & Set of output links whose flows have yet to be fully determined as of iteration $k$ \\ 
	$\St_i^c(k)$ & Adjusted demand for commodity $c$ of link $i$ as of iteration \\ 
	$\St_i(k)$ & Total adjusted demand for input link $i$ as of iteration $k$ \\ 
	$\beta_{i,j}^c$ & Split ratio of $c$ vehicles from link $i$ to link $j$ \\ 
	$\pt_i(k)$ & Adjusted priority of link $i$ at iteration $k$ \\ 
	$\Rt_j(k)$ & Adjusted supply of link $j$ at iteration $k$ \\ 
	$U_j(k)$ & Set of input links contributing to link $j$ whose flows are undetermined as of iteration $k$ \\ 
	$\pt_{i,j}(k)$ & Oriented priority from link $i$ to $j$ at iteration $k$ \\ 
	$a_j(k)$ & Restriction term of link $j$ at iteration $k$ \\ 
	$a_{j^\ast}(k)$ & Smallest (most restrictive) restriction term at iteration $k$ \\ 
	$\alpha_j(k)$ & Reduction factor of link $j$ at iteration $k$ \\ 
	$\Ut(k)$ & Set of input links whose demand can be fully met by downstream links at iteration $k$ \\ 
\end{longtable}

\section{MISO node solution algorithm}\label{app_merge}
\begin{enumerate}
\item Initialize:
\begin{eqnarray*}
\Rt_1(0) & := & R_1; \\
U(0) & := & \left\{1,\dots,M\right\}; \\
k & := & 0.
\end{eqnarray*}
Here, $k$ is the iteration index;
$\Rt_1(k)$ is the remaining supply of the output link at iteration $k$; and
$U(k)$ is the set of still unprocessed input links at iteration $k$:
input links whose input-output flows have not been assigned yet.

\item Check that at least one of the unprocessed input links has
nonzero priority, otherwise, assign equal positive priorities
to all the unprocessed input links:
\begin{equation*}
\pt_i(k) = \left\{\begin{array}{ll}
p_i, &
\mbox{ if there exists } i'\in U(k):\; p_{i'} > 0,\\
\frac{1}{|U(k)|}, &
\mbox{ otherwise},
\end{array}\right.
\end{equation*}
where $|U(k)|$ denotes the number of elements in set $U(k)$.

\item Define the set of input links
that want to send fewer vehicles than their allocated supply
and whose flows are still undetermined:
\begin{equation*}
\Ut(k) = \left\{i\in U(k):\; \sum_{c=1}^C S_i^c\leq\pt_i(k)
\frac{\Rt_1(k)}{\sum_{i'\in U(k)}\pt_{i'}(k)} \right\}.
\end{equation*}
\begin{itemize}
\item If $\Ut(k) \neq\emptyset$, assign:
\begin{eqnarray*}
f_{i1}^c & = & S_i^c, \;\;\; i\in\Ut(k); \\
\Rt_1(k+1) & = & \Rt_1(k) - \sum_{i\in\Ut(k)}\sum_{c=1}^C f_{i,1}^c; \\
U(k+1) & = & U(k) \setminus \Ut(k).
\end{eqnarray*}
\item Else, assign:
\begin{eqnarray*}
f_{i,1}^c & = & S_i^c \frac{\pt_i(k)}{\sum_{i'\in U(k)}\pt_{i'}(k)}
\frac{\Rt_1(k)}{S_i},
\;\;\; i\in U(k); \\
U(k+1) & = & \emptyset .
\end{eqnarray*}
\end{itemize}

\item If $U(k+1)=\emptyset$, then stop.

\item Set $k:=k+1$, and return to step 2.
\end{enumerate}

This algorithm finishes after no more than $M$ iterations, and
in the special case of $M=2$ it reduces to
\begin{eqnarray}
f_{1,1}^c & = & \min\left\{S_1^c, \;\; S_1^c\frac{\max\left\{
\frac{\pt_1}{\pt_1+\pt_2}R_1, \;\; R_1-S_2\right\}}{S_1}\right\};
\label{eq_merge_21_1}\\
f_{2,1}^c & = & \min\left\{S_2^c, \;\; S_2^c\frac{\max\left\{
\frac{\pt_2}{\pt_1+\pt_2}R_1, \;\; R_1-S_1\right\}}{S_2}\right\},
\label{eq_merge_21_2}
\end{eqnarray}
with $\pt_i$ computed per step 2.

\bibliographystyle{abbrvnat}
\bibliography{../../../traffic}

\providecommand{\url}[1]{#1}
\begin{thebibliography}{27}
\providecommand{\natexlab}[1]{#1}
\providecommand{\url}[1]{\texttt{#1}}
\expandafter\ifx\csname urlstyle\endcsname\relax
  \providecommand{\doi}[1]{doi: #1}\else
  \providecommand{\doi}{doi: \begingroup \urlstyle{rm}\Url}\fi

\bibitem[Belletti et~al.(2017)Belletti, Haziza, Gomes, and
  Bayen]{belletti2017expert}
F.~Belletti, D.~Haziza, G.~Gomes, and A.~M. Bayen.
\newblock Expert level control of ramp metering based on multi-task deep
  reinforcement learning.
\newblock \emph{arXiv preprint arXiv:1701.08832}, 2017.

\bibitem[Bliemer(2007)]{bliemer07}
M.~Bliemer.
\newblock Dynamic queueing and spillback in an analytical multiclass dynamic
  network loading model.
\newblock \emph{Transportation Research Record}, 2029:\penalty0 14--21, 2007.

\bibitem[{\relax California Department of Transportation
  (Caltrans)}(2015)]{csmps}
{\relax California Department of Transportation (Caltrans)}.
\newblock Office of {S}ystem, {F}reight, \& {R}ail {P}lanning, 2015.
\newblock URL \url{http://www.dot.ca.gov/hq/tpp/corridor-mobility/}.

\bibitem[Coogan and Arcak(2014)]{coogan2014freeway}
S.~Coogan and M.~Arcak.
\newblock Freeway traffic control from linear temporal logic specifications.
\newblock In \emph{ICCPS'14: ACM/IEEE 5th International Conference on
  Cyber-Physical Systems}, pages 36--47. IEEE Computer Society, 2014.

\bibitem[Corthout et~al.(2012)Corthout, Fl{\"o}tter{\"o}d, Viti, and
  Tamp\'{e}re]{corthout_non-unique_2012}
R.~Corthout, G.~Fl{\"o}tter{\"o}d, F.~Viti, and C.~M. Tamp\'{e}re.
\newblock Non-unique flows in macroscopic first-order intersection models.
\newblock \emph{Transportation Research Part B: Methodological}, 46\penalty0
  (3):\penalty0 343--359, Mar. 2012.
\newblock ISSN 01912615.
\newblock \doi{10.1016/j.trb.2011.10.011}.
\newblock URL
  \url{http://linkinghub.elsevier.com/retrieve/pii/S0191261511001652}.

\bibitem[Daganzo(1994)]{daganzo94}
C.~Daganzo.
\newblock The cell transmission model: A dynamic representation of highway
  traffic consistent with the hydrodynamic theory.
\newblock \emph{Transportation Research Part B: Methodological}, 28\penalty0
  (4):\penalty0 269--287, 1994.

\bibitem[Daganzo(1995)]{daganzo95a}
C.~Daganzo.
\newblock The cell transmission model, {Part II}: Network traffic.
\newblock \emph{Transportation Research Part B: Methodological}, 29\penalty0
  (2):\penalty0 79--93, 1995.
\newblock \doi{10.1016/0191-2615(94)00022-R}.

\bibitem[Fl\"{o}tter\"{o}d and Rohde(2011)]{flotterod_operational_2011}
G.~Fl\"{o}tter\"{o}d and J.~Rohde.
\newblock Operational macroscopic modeling of complex urban road intersections.
\newblock \emph{Transportation Research Part B: Methodological}, 45\penalty0
  (6):\penalty0 903--922, July 2011.
\newblock ISSN 01912615.
\newblock \doi{10.1016/j.trb.2011.04.001}.

\bibitem[Gentile et~al.(2007)Gentile, Meschini, and Papola]{gentile07}
G.~Gentile, L.~Meschini, and N.~Papola.
\newblock Spillback congestion in dynamic traffic assignment: a macroscopic
  flow model with time-varying bottlenecks.
\newblock \emph{Transportation Research Part B: Methodological}, 41\penalty0
  (10):\penalty0 1114--1138, 2007.

\bibitem[Gibb(2011)]{gibb_model_2011}
J.~Gibb.
\newblock Model of {Traffic} {Flow} {Capacity} {Constraint} {Through} {Nodes}
  for {Dynamic} {Network} {Loading} with {Queue} {Spillback}.
\newblock \emph{Transportation Research Record: Journal of the Transportation
  Research Board}, 2263:\penalty0 113--122, Dec. 2011.
\newblock ISSN 0361-1981.
\newblock \doi{10.3141/2263-13}.
\newblock URL \url{http://trrjournalonline.trb.org/doi/10.3141/2263-13}.

\bibitem[Gomes and Horowitz(2006)]{gomes06}
G.~Gomes and R.~Horowitz.
\newblock Optimal freeway ramp metering using the asymmetric cell transmission
  model.
\newblock \emph{Transportation Research, Part C}, 14\penalty0 (4):\penalty0
  244--262, 2006.

\bibitem[Hadi et~al.(2013)Hadi, Shabanian, Ozen, Xiao, Doherty, Segovia, and
  Ham]{floridaDOT2013}
M.~Hadi, S.~Shabanian, H.~Ozen, Y.~Xiao, M.~Doherty, C.~Segovia, and H.~Ham.
\newblock Application of {D}ynamic {T}raffic {A}ssignment to {A}dvanced
  {M}anaged {L}ane {M}odeling.
\newblock Technical report, Florida International University, 2013.

\bibitem[Jabari(2016)]{jabari_node_2016}
S.~E. Jabari.
\newblock Node modeling for congested urban road networks.
\newblock \emph{Transportation Research Part B: Methodological}, 91:\penalty0
  229--249, Sept. 2016.
\newblock ISSN 01912615.
\newblock \doi{10.1016/j.trb.2016.06.001}.

\bibitem[Lebacque and Khoshyaran(2005)]{lebacque05}
J.~Lebacque and M.~Khoshyaran.
\newblock First-order macroscopic traffic flow models: intersection modeling,
  network modeling.
\newblock In \emph{The 16th International Symposium on Transportation and
  Traffic Theory (ISTTT)}, pages 365--386, 2005.

\bibitem[Ni and Leonard(2005)]{ni_simplified_2005}
D.~Ni and J.~D. Leonard.
\newblock A simplified kinematic wave model at a merge bottleneck.
\newblock \emph{Applied Mathematical Modelling}, 29\penalty0 (11):\penalty0
  1054--1072, Nov. 2005.
\newblock ISSN 0307904X.
\newblock \doi{10.1016/j.apm.2005.02.008}.

\bibitem[Nie and Zhang(2005)]{nie_linkmodels_2005}
X.~Nie and H.~M. Zhang.
\newblock A comparative study of some macroscopic link models used in dynamic
  traffic assignment.
\newblock \emph{Networks and Spatial Economics}, 5\penalty0 (1):\penalty0
  89--115, 2005.

\bibitem[Papageorgiou et~al.(1991)Papageorgiou, Hadj-Salem, and
  Blosseville]{alinea91}
M.~Papageorgiou, H.~Hadj-Salem, and J.~Blosseville.
\newblock {ALINEA}: a local feedback control law for on-ramp metering.
\newblock \emph{Transportation Research Record}, 1320:\penalty0 58--64, 1991.

\bibitem[Raadsen et~al.(2016)Raadsen, Bliemer, and
  Bell]{raadsen_efficient_2016}
M.~P. Raadsen, M.~C. Bliemer, and M.~G. Bell.
\newblock An efficient and exact event-based algorithm for solving simplified
  first order dynamic network loading problems in continuous time.
\newblock \emph{Transportation Research Part B: Methodological}, 92:\penalty0
  191--210, Oct. 2016.
\newblock ISSN 01912615.
\newblock \doi{10.1016/j.trb.2015.08.004}.

\bibitem[Rezaee(2014)]{rezaee2014decentralized}
K.~Rezaee.
\newblock \emph{Decentralized coordinated optimal ramp metering using
  multi-agent reinforcement learning}.
\newblock PhD thesis, University of Toronto, 2014.

\bibitem[Shiomi et~al.(2015)Shiomi, Taniguchi, Uno, Shimamoto, and
  Nakamura]{shiomi2015}
Y.~Shiomi, T.~Taniguchi, N.~Uno, H.~Shimamoto, and T.~Nakamura.
\newblock Multilane first-order traffic flow model with endogenous
  representation of lane-flow equilibrium.
\newblock \emph{Transportation Research Part C: Emerging Technologies},
  59:\penalty0 198--215, Oct. 2015.
\newblock ISSN 0968090X.
\newblock \doi{10.1016/j.trc.2015.07.002}.
\newblock URL
  \url{http://linkinghub.elsevier.com/retrieve/pii/S0968090X15002429}.

\bibitem[Smits et~al.(2015)Smits, Bliemer, Pel, and van
  Arem]{smits_family_2015}
E.-S. Smits, M.~C. Bliemer, A.~J. Pel, and B.~van Arem.
\newblock A family of macroscopic node models.
\newblock \emph{Transportation Research Part B: Methodological}, 74:\penalty0
  20--39, Apr. 2015.
\newblock ISSN 01912615.
\newblock \doi{10.1016/j.trb.2015.01.002}.
\newblock URL
  \url{http://linkinghub.elsevier.com/retrieve/pii/S0191261515000053}.

\bibitem[Tamp\`{e}re et~al.(2011)Tamp\`{e}re, Corthout, Cattrysse, and
  Immers]{tampere11}
C.~M.~J. Tamp\`{e}re, R.~Corthout, D.~Cattrysse, and L.~H. Immers.
\newblock A generic class of first order node models for dynamic macroscopic
  simulation of traffic flows.
\newblock \emph{Transportation Research, Part B}, 45\penalty0 (1):\penalty0
  289--309, 2011.
\newblock \doi{10.1016/j.trb.2010.06.004}.

\bibitem[Work et~al.(2010)Work, Blandin, Tossavainen, Piccoli, and
  Bayen]{work2009trafficmodel}
D.~Work, S.~Blandin, O.-P. Tossavainen, B.~Piccoli, and A.~Bayen.
\newblock A traffic model for velocity data assimilation.
\newblock \emph{Applied \text{M}athematics \text{R}esearch e\text{X}press},
  2010\penalty0 (1):\penalty0 1--35, 2010.

\bibitem[Wright and Horowitz(2016)]{wright_pf_2016}
M.~Wright and R.~Horowitz.
\newblock Fusing {L}oop and {GPS} {P}robe {M}easurements to {E}stimate
  {F}reeway {D}ensity.
\newblock \emph{IEEE Transactions on Intelligent Transportation Systems},
  17\penalty0 (12):\penalty0 3577--3590, Dec 2016.
\newblock ISSN 1524-9050.
\newblock \doi{10.1109/TITS.2016.2565438}.

\bibitem[Wright et~al.(2016)Wright, Horowitz, and
  Kurzhanskiy]{wright_node_dynamic_2016}
M.~Wright, R.~Horowitz, and A.~A. Kurzhanskiy.
\newblock {A dynamic system characterization of road network node models}.
\newblock In \emph{Proceedings of the 10th IFAC Symposium on Nonlinear Control
  Systems}, volume~49, pages 1054--1059, August 2016.
\newblock \doi{10.1016/j.ifacol.2016.10.307}.

\bibitem[Wright and Horowitz(2017)]{wright2017generic}
M.~A. Wright and R.~Horowitz.
\newblock Generic second-order macroscopic traffic node model for general
  multi-input multi-output road junctions via a dynamic system approach.
\newblock \emph{arXiv preprint arXiv:1707.09346}, 2017.

\bibitem[Yperman et~al.(2005)Yperman, Logghe, and Immers]{yperman_link_2005}
I.~Yperman, S.~Logghe, and B.~Immers.
\newblock The link transmission model: {An} efficient implementation of the
  kinematic wave theory in traffic networks.
\newblock In \emph{Proceedings of the 10th {EWGT} {Meeting}}, pages 122--127,
  Poznan, Poland, 2005.

\end{thebibliography}

\end{document}